\documentclass[sigconf]{acmart}

\usepackage{subfig}
\usepackage{xspace}
\usepackage{enumitem}
\usepackage{amsfonts}
\usepackage[algo2e, ruled,vlined,linesnumbered]{algorithm2e} 

\usepackage[utf8]{inputenc}
\usepackage[english]{babel}

\newtheorem{theorem}{Theorem}[section]

\newtheorem{lemma}[theorem]{Lemma}

\newtheorem{assumption}[theorem]{Assumption}

\newcommand{\allreduce}{\emph{allreduce}\xspace}
\newcommand{\eg}{\emph{e.g.}\xspace}

\newcommand{\be}{\begin{equation}}
\newcommand{\ee}{\end{equation}}
\newcommand{\bee}{\begin{equation*}}
\newcommand{\eee}{\end{equation*}}
\newcommand{\bea}{\begin{eqnarray}}
\newcommand{\eea}{\end{eqnarray}}
\newcommand{\beaa}{\begin{eqnarray*}}
\newcommand{\eeaa}{\end{eqnarray*}}

\newcommand{\E}{\mathbb{E}}

\newcommand{\xt}{x^t}
\newcommand{\xte}{x^{t+1}}

\newcommand{\revise}[1]{{\color{black}#1}}

\AtBeginDocument{%
  \providecommand\BibTeX{{%
    \normalfont B\kern-0.5em{\scshape i\kern-0.25em b}\kern-0.8em\TeX}}}

\setcopyright{acmcopyright}

\copyrightyear{2020} 
\acmYear{2020} 
\setcopyright{usgovmixed}\acmConference[HPDC '20]{Proceedings of the 29th International Symposium on High-Performance Parallel and Distributed Computing}{June 23--26, 2020}{Stockholm, Sweden}
\acmBooktitle{Proceedings of the 29th International Symposium on High-Performance Parallel and Distributed Computing (HPDC '20), June 23--26, 2020, Stockholm, Sweden}
\acmPrice{15.00}
\acmDOI{10.1145/3369583.3392681}
\acmISBN{978-1-4503-7052-3/20/06}

\begin{document}
\fancyhead{}
\title{FFT-based Gradient Sparsification for the Distributed Training of Deep Neural Networks}



\author{Linnan Wang}
\affiliation{%
 \institution{Brown University}
 \city{Providence, RI}
 \country{USA}
}

\author{Wei Wu}
\affiliation{%
 \institution{Los Alamos National Laboratory}
 \city{Los Alamos, NM}
 \country{USA}
}

\author{Junyu Zhang}
\affiliation{
 \institution{University of Minnesota, Twin Cities}
 \city{Minneapolis, MN}
 \country{USA}
}

\author{Hang Liu}
\affiliation{%
 \institution{Stevens Institute of Technology}
 \city{Hoboken, NJ}
 \country{USA}
}

\author{George Bosilca}
\affiliation{%
 \institution{University of Tennessee}
 \city{Knoxville, TN}
 \country{USA}
}

\author{Maurice Herlihy}
\affiliation{%
 \institution{Brown University}
 \city{Providence, RI}
 \country{USA}
}

\author{Rodrigo Fonseca}
\affiliation{%
 \institution{Brown University}
 \city{Providence, RI}
 \country{USA}
}

\renewcommand{\shortauthors}{Wang and Wu, et al.}

\begin{abstract}
The performance and efficiency of distributed training of Deep Neural Networks (DNN) highly depend on the performance of gradient averaging among participating processes, a step bound by communication costs. 
There are two major approaches to reduce communication overhead: overlap communications with computations (lossless), or reduce communications (lossy). The lossless solution works well for linear neural architectures, e.g. VGG, AlexNet, but more recent networks such as ResNet and Inception limit the opportunity for such overlapping. Therefore, approaches that reduce the amount of data (lossy) become more suitable. 
In this paper, we present a novel, explainable lossy method that sparsifies gradients in the frequency domain, in addition to a new range-based float point representation to quantize and further compress gradients. These dynamic techniques strike a balance between compression ratio, accuracy, and computational overhead, and are  optimized to maximize performance in heterogeneous environments. 

Unlike existing works that strive for a higher compression ratio, we stress the robustness of our methods, and provide guidance to recover accuracy from failures. To achieve this, we prove how the FFT sparsification affects the convergence and accuracy, and show that our method is guaranteed to converge using a diminishing $\theta$ in training. Reducing $\theta$ can also be used to recover accuracy from the failure. Compared to STOA lossy methods, e.g., QSGD, TernGrad, and Top-k sparsification, our approach incurs less approximation error, thereby better in both the wall-time and accuracy. On an 8 GPUs, InfiniBand interconnected cluster, our techniques effectively accelerate AlexNet training up to 2.26x to the baseline of no compression, and 1.31x to QSGD, 1.25x to Terngrad and 1.47x to Top-K sparsification.
\end{abstract}



\maketitle

\section{Introduction}
Parameter Server (PS) and allreduce-style communications are two core parallelization strategies for distributed DNN training. In an iteration, each worker produces a gradient, and both parallelization strategies rely on the communication network to average the gradients across all workers. The gradient size of current DNNs is at the scale of $10^2$ MB, and, even with the state-of-the-art networks such as Infiniband, repeatedly transferring such a large volume of messages over millions of iterations is prohibitively expensive. Furthermore, the tremendous improvement in GPU computing and memory speeds (e.g., the latest NVIDIA TESLA V100 GPU features a peak performance of 14 TFlops on single-precision and memory bandwidth of 900 GB/s with HBM2) further underscores communication as a bottleneck.


Recently, several methods have shown that training can be done with a lossy gradient due to the iterative nature of Stochastic Gradient Descent (SGD). It opens up new opportunities to alleviate the communication overhead by aggressively compressing gradients. One approach to compress the gradients is \emph{quantization}. For example, Terngrad~\cite{wen2017terngrad} maps a gradient into [-1, 0, 1], and QSGD~\cite{alistarh2017qsgd} stochastically quantizes gradients onto a uniformly discretized set larger than that of Terngrad. Such coarse approximation not only incurs large errors between the actual and quantized gradients as we demonstrate in Figure~\ref{gradient_dist_alg_comps} [QSGD, TernGrad], but also fails to exploit the bit efficiency in the quantization (Figure~\ref{quantization_scheme}). 
Another approach to gradient compression, \emph{sparsification}, only keeps the top-k largest gradients \cite{han2015deep,aji2017sparse,alistarh2018convergence}. Similarly, Top-k loses a significant amount of actual gradients around zeros to achieve a high compression ratio (Figure~\ref{gradient_dist_alg_comps}, [Top-k]). In summary, existing lossy methods greatly drop gradients, incur large approximation errors (Figure~\ref{recon_error}), leading to the deterioration of the final accuracy (Table~\ref{algorithm_results}). To avoid compromising the convergence speed, both \textit{quantization} and \textit{sparsification} must limit the compression ratio, leading to sub-optimal improvement of the end-to-end training wall time.

In this paper, we propose a gradient compression framework that takes advantages of both $sparsification$ and $quantization$ with two novel components, FFT-based sparsification, and a range-based quantization. FFT-based sparsification allows removing the redundant information, while preserving the most relevant information (Figure~\ref{gradient_dist_alg_comps} [FFT]). As a result, FFT incurs fewer errors in approximating the actual gradients (Figure~\ref{recon_error}), thereby better in accuracy than QSGD, TernGrad, and Top-K (Table~\ref{algorithm_results}). We treat the gradient as a 1D signal, and drop near-zero coefficients in the frequency domain, after an FFT. Deleting some frequency components after the FFT introduces magnitude errors, but the signal maintains its distribution (Figure~\ref{FFT_Time_Sparsification}). As a result, the sparsification in the frequency domain can achieve the same compression ratio as in the spatial domain but preserving more relevant information.

To further improve the end-to-end training wall time, we introduce a new range-based variable precision floating point representation to quantize and compress the gradient frequencies after sparsification. Most importantly, unlike the uniform quantization used in existing approaches, the precision of representable floats in our method can be adjusted to follow the distribution of the original gradients (Figure~\ref{adjustable_representation_range}). The novel range-based design allows us to fully exploit the precision given limited bits so that the approximation error can be further reduced. By combining \textit{sparsification} and \textit{quantization}, our framework delivers a higher compression ratio than the single method, resulting in shorter end-to-end training wall time than QSGD, Terngrad, and Top-k.
 
Lastly, our compression framework is highly efficient and scalable. The primitive algorithms in our compression scheme, such as FFT, top-k select, and precision conversions, are efficiently parallelizable and thus GPU-friendly. We resort to existing highly optimized GPU libraries such as cuFFT, Thrust, and bucketSelect \cite{alabi2012fast}, while we propose a simple yet efficient packing algorithm to transform sparse gradients into a dense representation. Minimizing the computational cost of the compression is crucial for high-speed networks, such as Infiniband networks, as we analyzed in Figure~\ref{primtives_comp}. 

Specifically, the contributions of this paper are as follows:
\begin{itemize}[itemsep=.5ex,leftmargin=3ex]
  \item a novel FFT-based, tunable gradient sparsification that retains the original gradient distribution.
  \item a novel range-based variable precision floating-point that allocates precision according to the gradient distribution.
  \item \revise{a analytic model to guide people when to enable compression and how to set a compression ratio according to hardware specifications.}
  \item the convergence proof of our methods, and its guidance in selecting a compression ratio $\theta$, to ensure the convergence, or reduce $\theta$ to recover the accuracy. \revise{To the best of our knowledge, this paper is the first one to discuss the relationship between compression ratio and accuracy of neural networks.}
  \item highly optimized system components for a compression framework that achieves high throughput on GPUs and is beneficial even on state-of-the-art Infiniband networks.
\end{itemize}

\section{Background and Motivation} \label{motivation}

\begin{figure}[!t]
    \centering
    \subfloat[][BSP\label{mpi_integration}]{\includegraphics[width=0.45\linewidth]{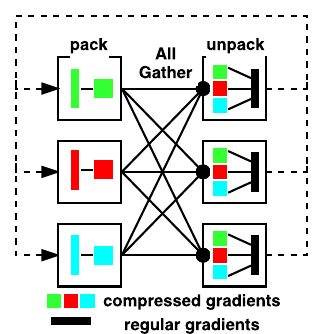}} \quad
    \subfloat[][PS\label{ps_integration}]{\includegraphics[width=0.40\linewidth]{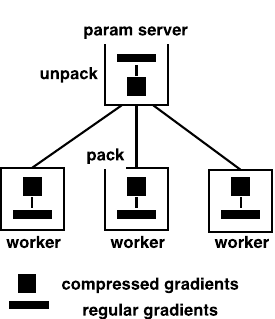} }
    \caption{\textbf{Two parallelization schemes of distributed DNN training}:(a) Bulk Synchronous Parallel (BSP) strictly synchronizes gradients with all-to-all group communications, e.g. MPI collectives; (b) Parameter Server (PS) exchanges gradients with point-to-point communications, e.g. push/pull.}
	\label{parallelization_scheme}
	\vspace{-0.4cm}
\end{figure}


In general, there are two strategies to parallelize DNN training: \textit{Model Parallelism} and \textit{Data Parallelism}. \textit{Model Parallelism} splits a network into several parts, with each being assigned to a computing node~\cite{dean2012large}. It demands extensive intra-DNN communications in addition to gradient exchanges. It largely restricts the training performance, and thereby \textit{Model Parallelism} is often applied in scenarios where the DNN cannot fit onto a computing node~\cite{dean2012large}. The second approach, \textit{Data Parallelism}~\cite{wang2017accelerating}, partitions the image batch, and every computing node holds a replica of the network. In a training iteration, a node computes a sub-gradient with a batch partition. Then, nodes \emph{all-reduce} sub-gradients to reconstruct the global one. The only communications are for necessary gradient exchanges. Therefore, current Deep Learning (DL) frameworks such as SuperNeurons~\cite{Wang:2018:SDG:3178487.3178491}, MXNet~\cite{chen2015mxnet}, Caffe~\cite{jia2014caffe}, and TensorFlow~\cite{abadi2016tensorflow} parallelize the training with \textit{Data Parallelism} for the high-performance.

There are two common strategies to organize the communications with data parallelism: with a centralized Parameter Server (PS) (Figure~\ref{ps_integration}), or with all-to-all group communications, e.g., \allreduce (Figure~\ref{mpi_integration}). TensorFlow~\cite{abadi2016tensorflow}, MXNet~\cite{chen2015mxnet}, and PaddlePaddle implement distributed DNN training with a Parameter Server (PS)~\cite{li2014scaling}.
In this distributed framework, the parameter server centralizes the parameter updates, while workers focus on computing gradients. Each worker pushes newly computed gradients to the parameter server, and the parameter server updates parameters before sending the latest parameters back to workers. Though this client-server~\cite{berson1992client} style design easily supports fault tolerance and elastic scalability, the major downside is the network congestion on the server. Alternatively, \allreduce-based Bulk Synchronous Parallel SGD can better exploit the bandwidth of a high-speed, dense interconnects, such as modern Infiniband networks. Instead of using a star topology, \allreduce pipelines the message exchanges at a fine granularity with adjacent neighbors in a ring-based topology. Since the pipeline fully utilizes the inbound and outbound link of every computing node, it maximizes network bandwidth utilization and achieves appealing scalability where the cost is largely independent of the number of computing nodes.

There are trade-offs between the BSP and PS schemes, with PS having better fault tolerance, and \allreduce better exploits the network bandwidth.
However, as we argue below, in both cases, the communication cost is high, and reducing it can yield substantial gains in training latency.

\subsection{Communication Challenges in Distributed Training of DNNs}

Communications for averaging sub-gradients is widely recognized as a major bottleneck in scaling DNN training\cite{zhao2017efficient, dean2012large, wang2017accelerating}. With increasing data complexity and volume, and with emerging non-linear neural architectures, two critical issues exacerbate the impact of communications in the scalability and efficiency of distributed DNN training with data parallelism: I) \textit{the increasing amounts of data to be exchanged}, and II) \textit{the decreasing opportunity to overlap computation and communication}.

\begin{figure}[!t]
\centering
        \subfloat[][AlexNet\label{AlexNet_layer_comm_comp}]{\includegraphics[height=2.4cm]{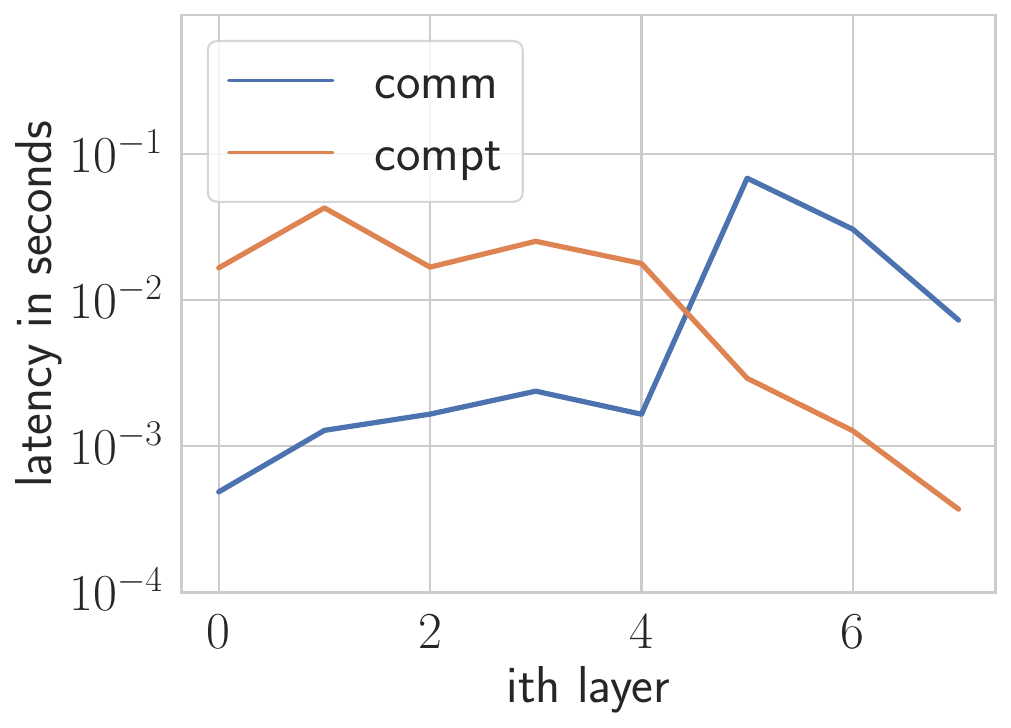}} \quad
        \subfloat[ResNet32\label{ResNet_layer_comm_comp}]{\includegraphics[height=2.4cm]{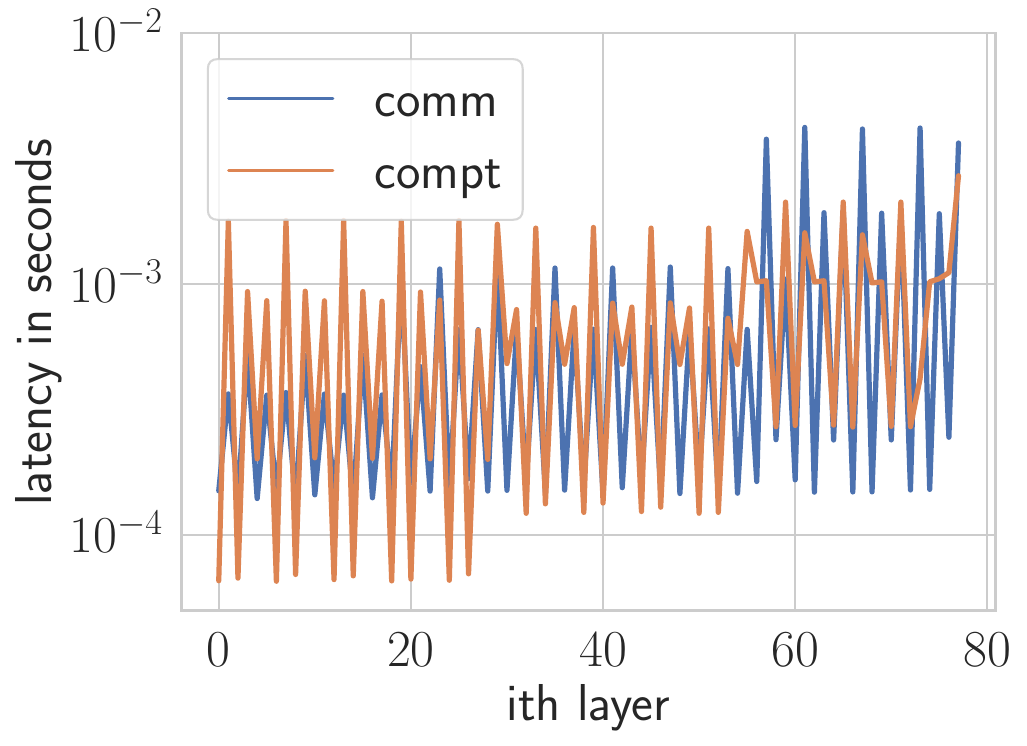}} 
\caption{\textbf{layer-wise communications (all-reduce) v.s. computations} in an iteration of BSP SGD using 16 P100 (4 GPUs/node with 56Gbps FDR). }
\vspace{-0.5cm}
\end{figure}

\noindent\textbf{Challenge I: \revise{Enormous amounts of communications} during training.}
DNNs are extremely effective at modeling complex nonlinearities thanks to the representation power of millions of parameters. The number of parameters dictates the size of the gradients. Specifically, the gradient sizes of AlexNet, VGG16, ResNet32, and Inception-V4 are 250MB, 553MB, 102MB, and 170MB. Even with the highly optimized allreduce implementation on a 56 Gbps FDR network, communication overhead remains significant. For example, the communication for AlexNet, VGG16, Inception-V4 and ResNet32 at regular single-GPU batch sizes\footnote{the single GPU batch size for AlexNet is 64, and 16 for others.} consumes $64.17\%, \linebreak
18.62\%, 33.07\%$ and $43.96\%$ of an iteration time, respectively.

\noindent\textbf{Challenge II: Decreasing opportunity to overlap computation and communication.}
One promising solution to alleviate the communication overhead is hiding the communication for averaging the gradient of the $i^{th}$ layer by the computation of $i-1^{th}$ layer in the backward pass. This lossless technique has proven to be effective on linear networks such as AlexNet and VGG16~\cite{awan2017s, rhu2016vdnn}, as these networks utilize large convolution kernels to process input data. Figure~\ref{AlexNet_layer_comm_comp} demonstrates the computation time of the convolution layers is $10\times$ larger than the communication time, easy for overlapping. 
\revise{
However, the overlapping technique is not always applicable for two reasons. 
$First$, the degree of overlapping is largely decided by the computation pattern of the neural network model. The opportunity for computation and communication overlap is very limited in recent neural architectures, such as Inception-V4~\cite{szegedy2017inception} and ResNet~\cite{he2016deep}.
The sparse fan-out connections in the Inception Unit (Figure 1a in \cite{Wang:2018:SDG:3178487.3178491}) replace one large convolution (e.g. 11$\times$11 convolution kernel in AlexNet) with several small convolutions (e.g. 3$\times$3 convolution kernels). Similarly, ResNet utilizes either 1$\times$1 or 3$\times$3 small convolution kernels. As a result, the layer-wise computational cost of ResNet is similar to or smaller than communication (Figure~\ref{ResNet_layer_comm_comp}); hence, it is much harder to overlap these neural networks than AlexNet.
$Second$, the degree of overlapping is also impacted by the bandwidth of networks. With slower networks, there are less opportunity to overlap communications and
computations. Specifically, as seen
in Figure~\ref{AlexNet_layer_comm_comp}, the computation cost of convolution layers of the AlexNet is $10\times$ larger than the communication cost with 56Gbps InfiniBand. 
However, when training AlexNet in a low profile network such as 1Gbps Ethernet, it becomes impossible to hide the communication cost as it is significantly larger than the computation cost. 
}

\begin{figure*}[t]
    \centering
        \includegraphics[width=0.95\linewidth]{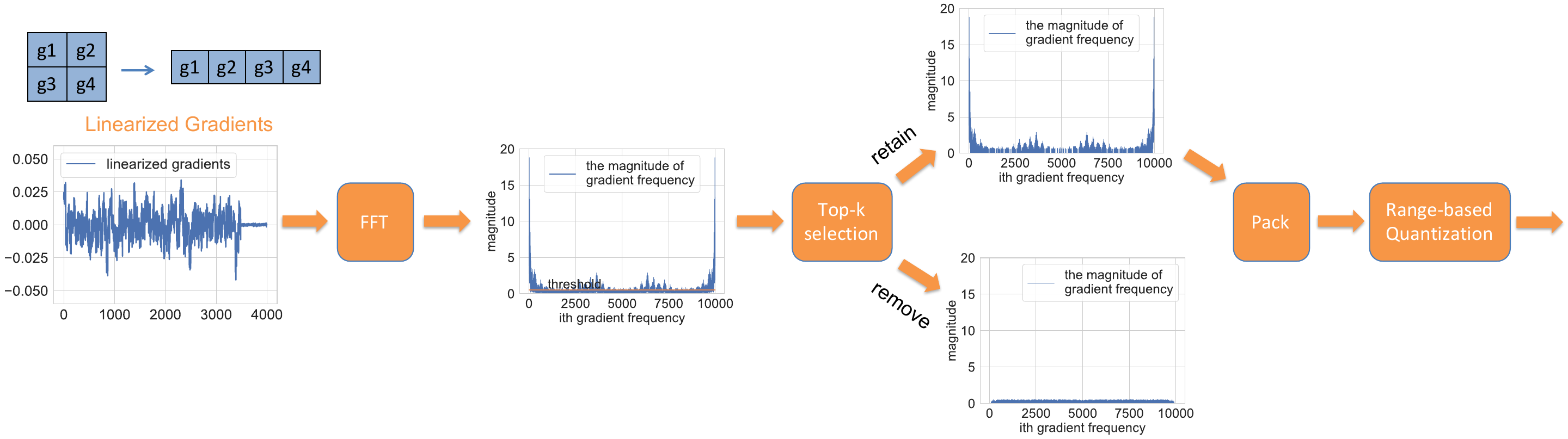}
        \caption{The gradient compression framework (sender).}
        \label{compression_pipeline}
        \vspace{-0.5cm}
        \hspace{-2cm}
\end{figure*}

These two challenges -- increasing data exchanged, and decreasing opportunity to hide communication latency -- make it attractive to look for solutions that minimize the communication cost by decreasing the communication volume.
Training a neural network with imprecise gradient updates still works as parameters are iteratively refined \cite{aji2017sparse}. Particularly, lossy gradient compression can achieve higher compression rates and still allow the network to deliver target accuracy \cite{alistarh2017qsgd}.
Given this, it is not surprising that several gradient compression approaches have been proposed in the literature. They generally fall into two categories: quantization of the gradients (\eg~\cite{seide20141, wen2017terngrad, alistarh2017qsgd, de2015taming}), where these are represented with lower precision numbers, and sparsification (\eg~\cite{aji2017sparse, alistarh2018convergence, wangni2018gradient}), where small gradient are treated as zero and not transmitted. We discuss these approaches in detail in Section~\ref{related_work}. As we describe next, we propose a novel gradient compression scheme that uses adaptive quantization and tunable FFT-based gradient compression that, together, achieve variable compression ratios that can maintain convergence quality, and, critically, is cheap enough computationally to be beneficial.

\section{Methodology}
\label{methodologies}

\subsection{The Compression Framework}
Figure~\ref{compression_pipeline} provides a step-by-step illustration of our compression pipeline. 
\begin{itemize}
  \item[1] Linearize the gradients by re-arranging gradient tensors into a 1-d vector for Fast Fourier Transform (\textbf{FFT}), which is discussed in Section~\ref{sec:stepfft}.
  \item[2] Truncate the gradient frequencies based on their magnitudes to sift out the \textbf{top-k} low-energy frequency components, which is discussed in Section~\ref{sec:stepfft}.
  \item[3] Transform the frequencies' representation from 32-bit float to a new, \textbf{range-based}, $N$-bit float ($N < 32$) to further compress down the gradient frequency,
  which is discussed in Section~\ref{range_base_8bit}.
  \item[4] \textbf{Pack} sparse data into dense vector and transfer them out, which is discussed in Section~\ref{sec:stepfft}.
\end{itemize}
On the receiver side, a similar approach (but using the inverse operations in the reverse order) is used to decompress the gradient frequency vector into gradients. Detailed discussions of compression components and their motivations are as follows. 

\begin{figure}[!t]
    \centering
    \subfloat[][CIFAR-10,  ResNet]{\includegraphics[height=2.5cm]{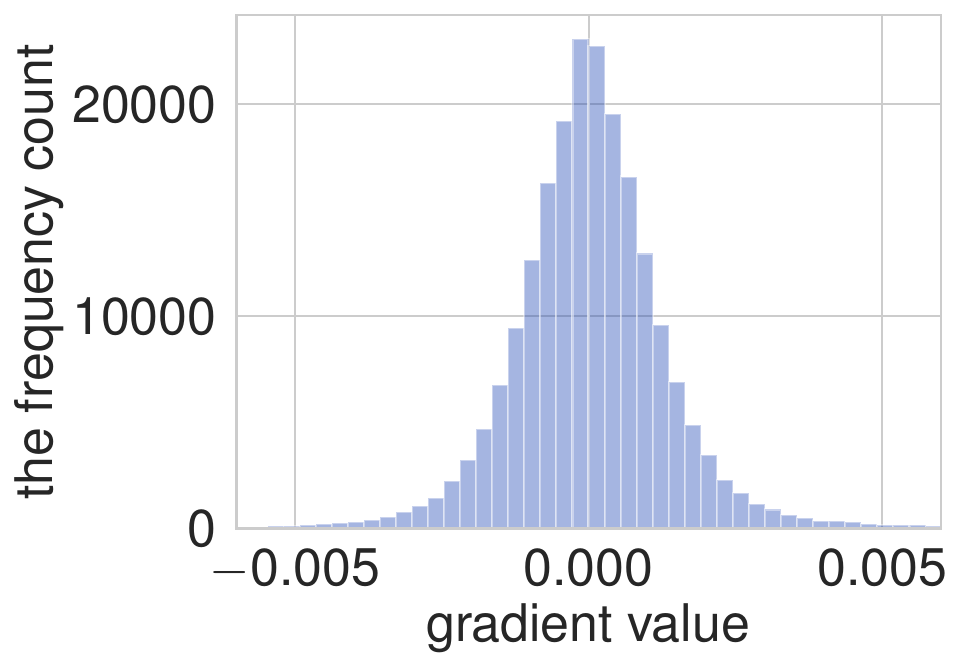}} \quad
    \subfloat[ImageNet, AlexNet]{\includegraphics[height=2.5cm]{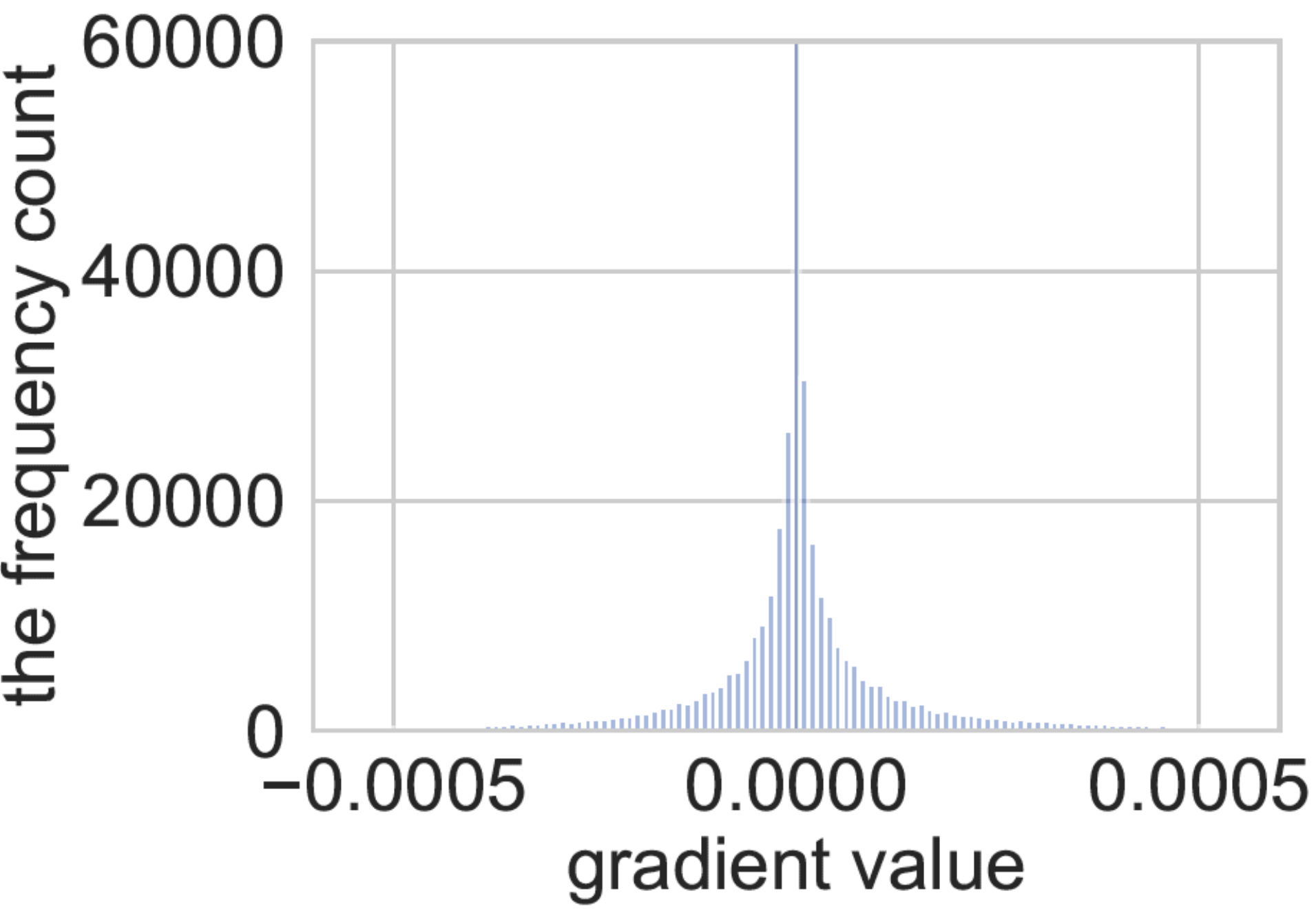}}
\caption{\textbf{Histogram of DNN gradients}: we sampled gradients every $10^3$ and $10^4$ iterations in a full training.}
\label{grad_dist}
\vspace{-0.5cm}
\end{figure}

\begin{figure}[!t]
    \centering
    \subfloat[][FFT Top-k\label{ftt_comp_signal}]{\includegraphics[height=2.5cm]{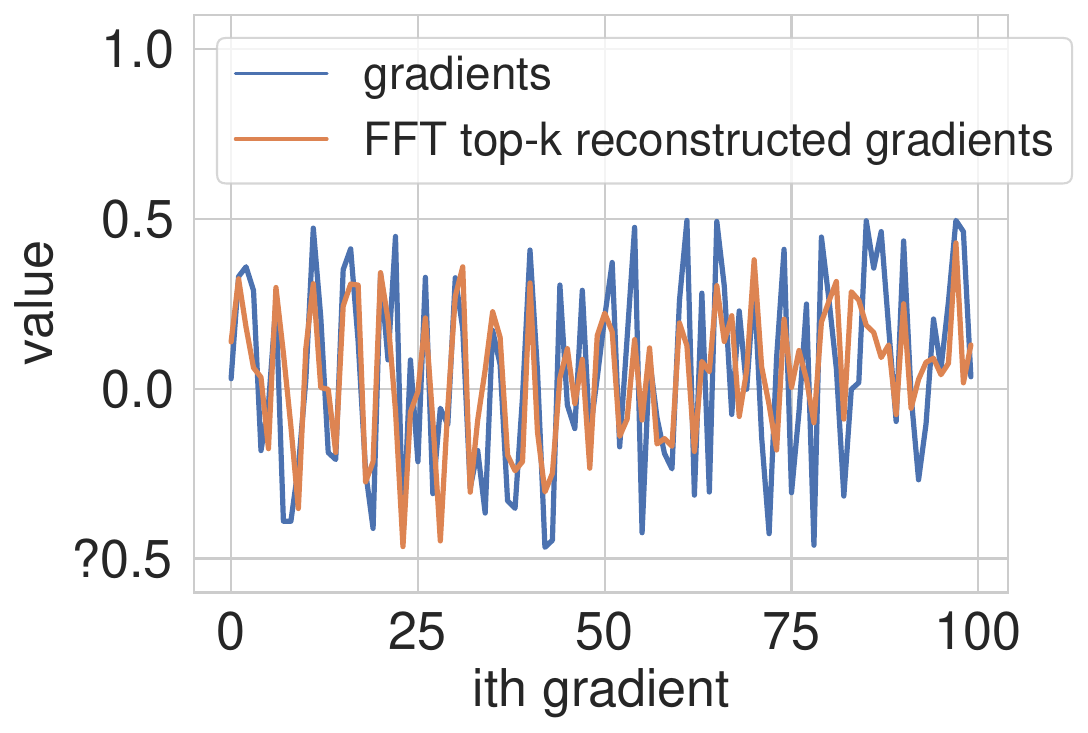}} \quad
    \subfloat[Top-k\label{time_comm_signal}]{\includegraphics[height=2.5cm]{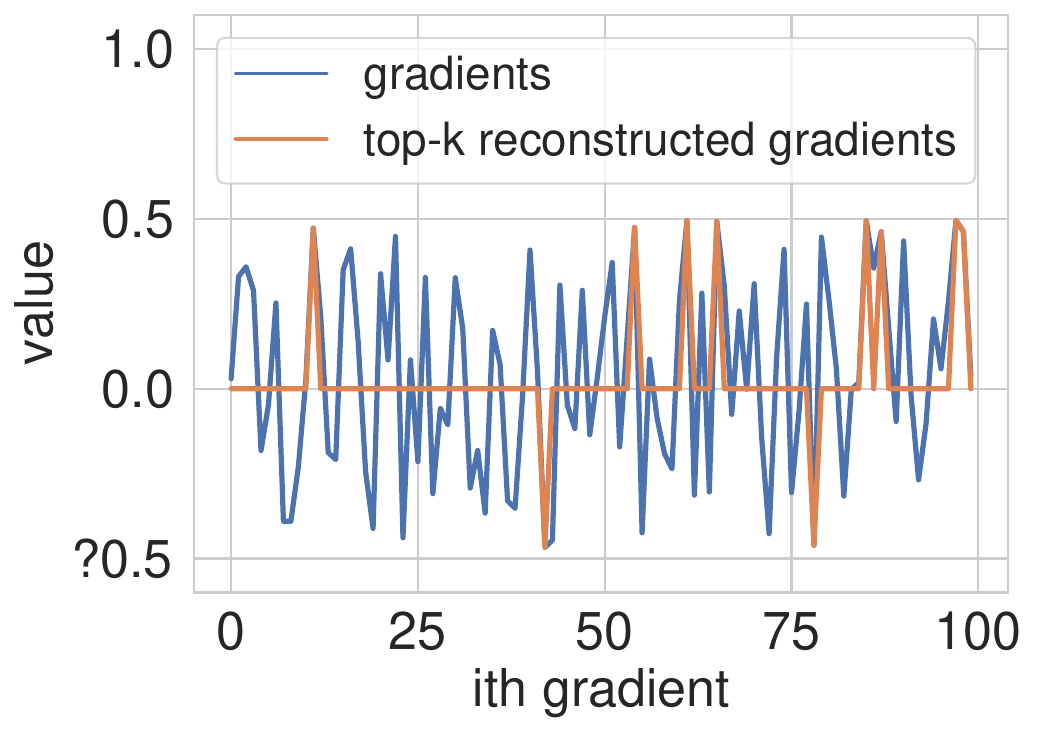}}
    \caption{ \textbf{FFT Top-k v.s. direct Top-k sparsificaiton}: Top-k aggressively loses gradients (err=0.0246), while FFT preserves more relevant information (err=0.0209) at the same sparsification ratio.}
    \label{FFT_Time_Sparsification}
    \vspace{-0.5cm}
\end{figure}

\subsubsection{Removing redundant information with FFT based Top-K sparsification} \label{sec:stepfft}
\textbf{Motivation}: the gradient points to a descent direction in the high dimensional space, thereby small perturbations on gradients can be viewed as introducing local deviations along the descent direction. If such deviations are limited during the training, these imprecise descent directions still iteratively lead to a local optimum at the cost of additional iterations. This is the intuition for the gradient sparsification. Besides, Figure~\ref{grad_dist} indicates high redundancy in DNN gradients due to a lot of near-zero components, that may have limited contributions in updating gradients. Recently, several top-k based methods~\cite{han2015deep,aji2017sparse,alistarh2018convergence} have also shown the possibility to train DNNs with only the top 10\% largest gradients. However, the resulting gradients, as shown in Figure~\ref{FFT_Time_Sparsification}, significantly deviate from the original, for entirely dropping the gradients below the threshold. This has motivated us to sparsify gradients, instead, in the frequency domain for preserving the trend of the original signal even after removing the same amount of information. For a gradient vector of length N, each gradients is $g_i = \sum_{n=0}^{N-1} x_ne^{\frac{-i2\pi k n}{N}}$ after FFT. If we sparsify on $x_n$, i.e. $g_i = \sum_{n}^{top_k} x_ne^{\frac{-i2\pi k n}{N}}$, $g_i$ still preserves some of the original gradient information. Therefore, FFT based top-k shows better results than top-k in Figure~\ref{FFT_Time_Sparsification}. More validations are available in the experimental section.

\textbf{Our approach}: The detailed computation steps of our FFT sparsification are highlighted in Figure~\ref{compression_pipeline}. Recent generations of NVIDIA GPUs support mixed-precision; and computing with half-precision increases the FFT throughput up to 2$\times$. So, we convert 32-bit (full-precision) gradients into 16-bit (half-precision) gradients to improve the throughput before applying FFT, and the information loss from the conversion is negligible due to the bounded gradients. 

After FFT, the next step is to filter the low energy gradient in the frequency domain. We introduce a new hyper-parameter, $\theta$, to regulate the sparsity of frequencies. Here, we only describe the procedures, and the tuning of $\theta$ is thoroughly discussed in Section~\ref{diminishing_lr_theorem} and experiments. If $\theta = 0.9$, we keep the top 10\% frequency components in magnitude and drop the rest by resetting to zeros (Figure~\ref{compression_pipeline}). The selection is implemented with either sorting or Top-k. Since Thrust\footnote{\url{https://developer.nvidia.com/thrust}} and cuFFT\footnote{\url{https://developer.nvidia.com/cufft}} provide highly optimized FFT and sorting kernels for the GPU architecture, we adopted them in our implementations.

\subsection{Packing sparse data into a dense vector}
Thresholding gradient frequencies in the last step yields a highly irregular sparse vector, and we need to pack it into a dense vector to reduce communications. The speed of packing a sparse vector is critical to the practical performance gain. Here, we propose a simple parallel packing algorithm:
\begin{itemize}
    \item[1] Create a $status$ vector and mark an element in $status$ as 1 if the corresponding scalar in $sparse$ vector is non-zero (e.g., $sparse = [a,0,b,0,c,0,0]$ and $status = [1, 0, 1, 0, 1, 0, 0]$).
    \item[2] Perform a parallel prefix-sum on $status$ to generate a $location$ vector ($[1, 1, 2, 2, 3, 3, 3]$).
    \item[3] if $status[i] == 1$, write $sparse[i]$ to $dense[location[i]]$, and $dense$ vector is the packed result.
\end{itemize}
This parallel algorithm has a 689$\times$ speedup over the single-threaded algorithm on a TESLA V100 with a throughput of 34 GB/s.

\begin{figure}[!t]
\centering
\includegraphics[width=0.75\linewidth]{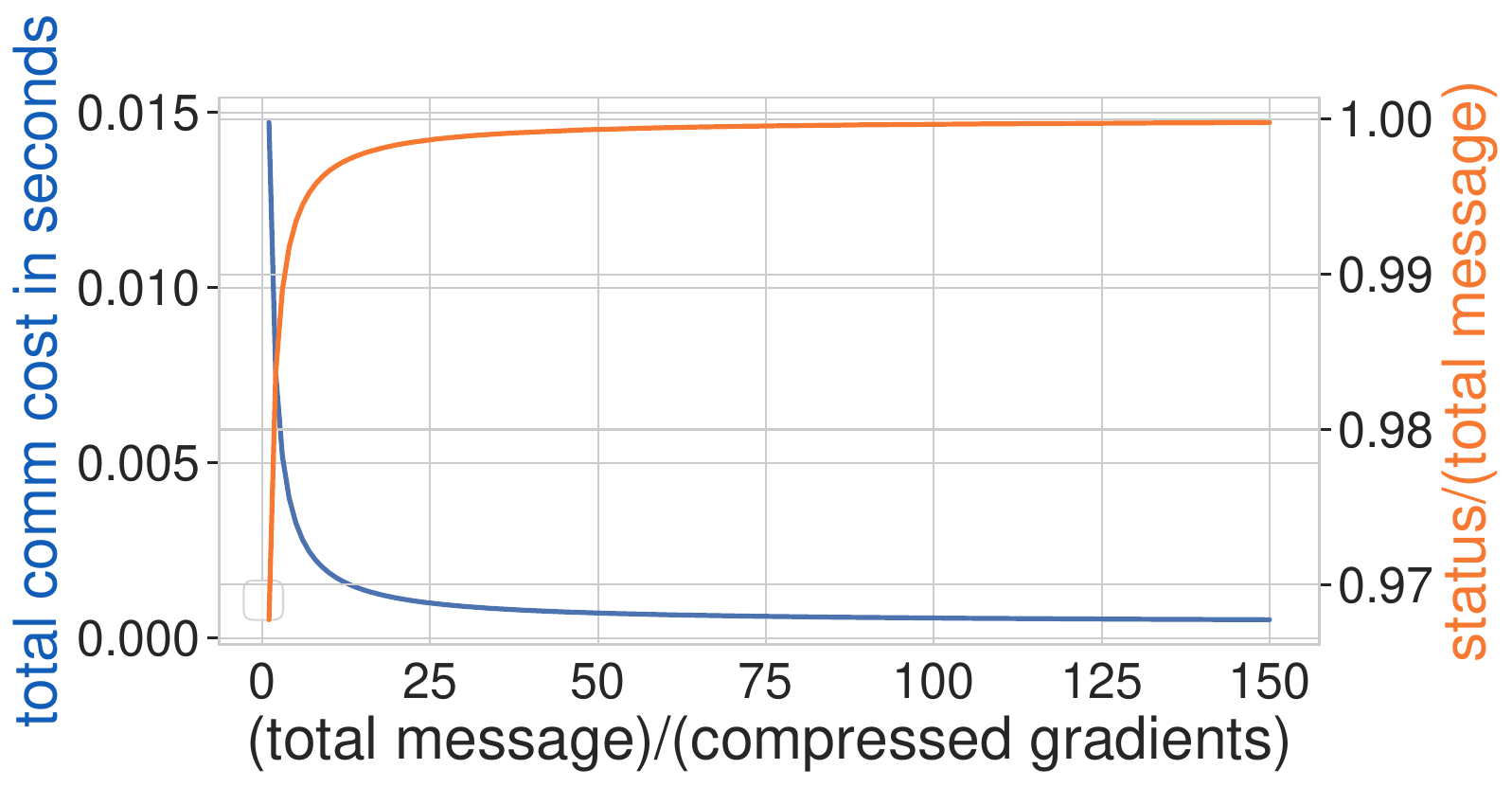}
\caption{\textbf{the effect of status vector}: given 100 MB gradients, the improvement after dropping $> 95\%$ gradients ($\theta=0.05$, compression ration is 20) is limited. }
\label{status_info}
\end{figure}

We need to send the status vector and the compressed gradient to perform the decompression. The status vector is a bitmap that tracks the location of non-zero elements, and its length in bits is the same as the gradient vector. Figure~\ref{status_info} shows the cost of the status vector is non-negligible after the compression ratio exceeding 20. Therefore, setting $\theta < 0.05$ is not desired.

\subsubsection{Range based Quantization} \label{range_base_8bit}

\begin{figure}[!t]
    \centering
	\includegraphics[height=3cm]{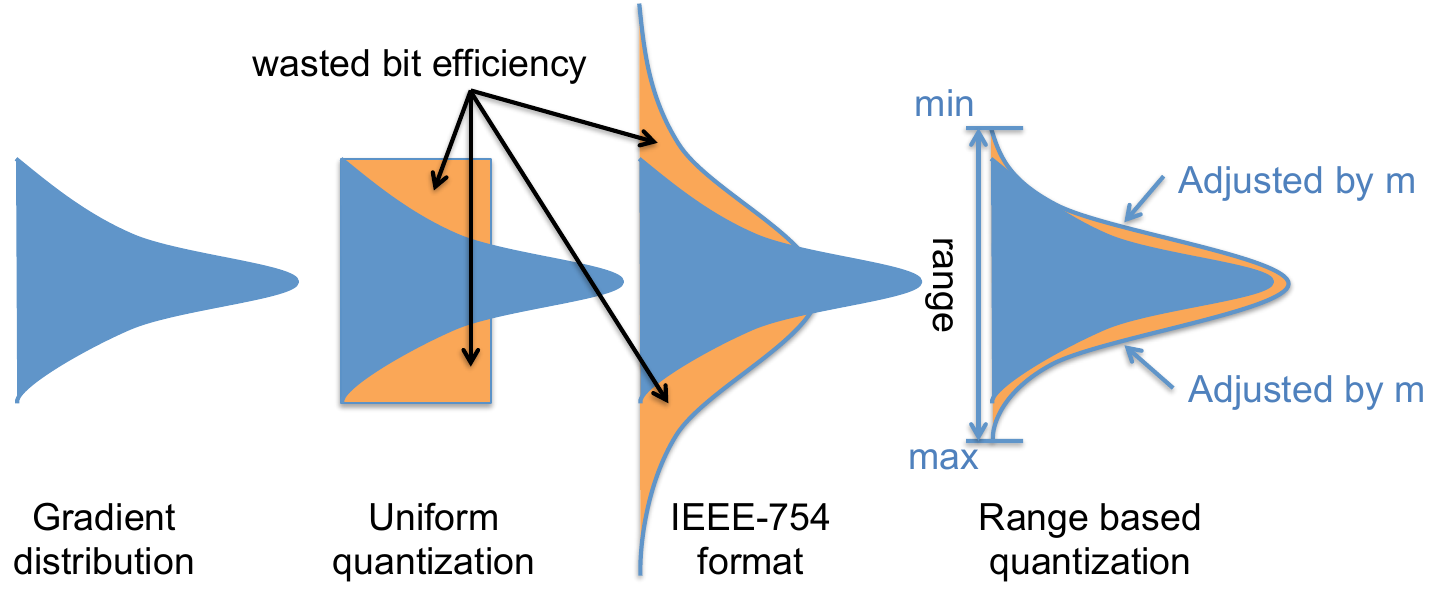} 
    \caption{ \textbf{comparisons of quantization schemes}: the uniform distribution and IEEE 754 format.}
	\label{quantization_scheme}
	\vspace{-0.5cm}
\end{figure}

\textbf{Motivation}: the range of single precision IEEE-754 floating point is $[-3.4*10^{38}, +3.4*10^{38}]$, while the range of gradients and their frequencies are much smaller (e.g. [-1, +1]). This motivates us to represent the bounded gradients with fewer bits.  The problem of using an N bits IEEE 754 format, as seen in Figure~\ref{quantization_scheme}, is the inconsistency between the range of gradients [$min$, $max$] and the range of the IEEE representable numbers. Given N bits for IEEE 754, there are $N-2$ combinations of exponent-mantissa. The representation range is either too large or too small for gradients, regardless of which combinations to choose. Another conventional way is to equally divide the $max-min$ into $2^N$, i.e., uniform quantization. Still, the actual gradient distribution is far from the uniform, and thereby it is also inefficient, as shown in Figure~\ref{quantization_scheme}.

\textbf{Our approach}: we propose an offset-based N-bit floating point, which intends to match the distribution of representable numbers to the real gradients. 
Our representation is to use the N-bit binary format of a positive number as base number $pbase$, and encode it to 0...01. The rest positive numbers are encoded as 0...01 ($pbase$) + offset. The negative numbers also follow the same rule. Therefore, the total $2^N$ representable numbers consist of $P$ positive numbers and $2^N-P$ negative numbers. To match the range of real gradients, our quantization permits the manual setting of a representation range, defined by $min$ and $max$. We estimate $min$ and $max$ from the first few iterations of gradients. Then, we tune $m$ and $eps$ to adjust the precision of representable numbers, as shown in Figure~\ref{quantization_scheme}. $m$ represents the number of bits left for the mantissa, and $eps$ represents the minimal representable positive number whose corresponding N-bit binary is $pbase$. The following further explains how $m$ and $eps$ adjust the precision:

\begin{itemize}
\item $m$: let's denote the difference between two consecutive numbers as $diff$. For $m$ bits mantissa, the exponent increases by 1 after $2^m$ number, and increasing $diff = diff*2$. Since $diff$ is exponentially growing, this creates a Gaussian like representation range that matches to real gradients. If $max$, $min$ and $eps$ are fixed, $P$ is small for a small $m$, as it takes fewer numbers to increase the exponent. Similarly, a large $m$ leads to a larger $P$. Therefore, $m$ is very sensitive for precision.
\item $eps$: with $max$, $min$ and $m$, $diff$ is also fixed. If $eps$ is small, it takes more steps to reach $max$ yielding a large $P$; and vice versa. 
\end{itemize}

Since $m$ and $eps$ determine $P$, we need to tune them to make $P$ close to $2^N/2$ for balancing the range of positive and negative numbers. In practice, $N$, $min$, and $max$ are empirically decided from gradients, and the $m \in [1, N]$. We iterate every $m$ to tune for eps. Given $N$, $m$, $min$, and $max$, we initialize $eps$ as a reasonably small number, e.g., 0.002, then de-compress the 1..1 (the minimal representable negative number) back to FP32 with the selected $eps$,
and the resulting number is the current actual minimal negative number $actual\_min$;
if $actual\_min$ is smaller than $min$, we decrease $eps$, and increase otherwise. Following this path, $P$ converges to $2^N/2$, a state with equal positive and negative numbers, and yielding the optimal $eps$.

\begin{figure}[!t]
    \centering
    \includegraphics[width=0.85\columnwidth]{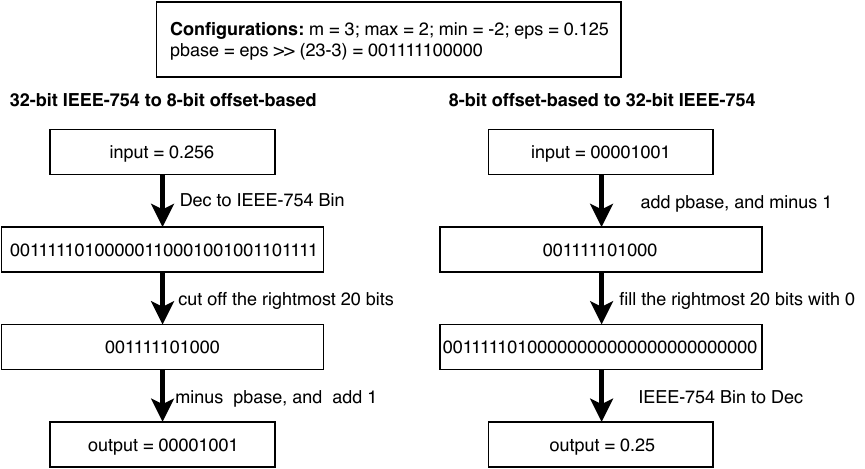}
    \caption{\textbf{\textit{Illustration of range based quantizer}}: an example conversion of between 32 bits IEEE 754 and 8 bits our representation. }
    \label{8bit_example}
    \vspace{-0.5cm}
\end{figure}

\begin{algorithm2e}[h]
\scriptsize
    \KwIn{init(min, max)} 
    pbase\_binary = eps >> (23-m) \; 
    \KwIn{32bit\_to\_Nbit(32bit\_float)}
    \If{32bit\_float > max} {
        32bit\_float = max;
    }
    32bit\_binary = 32bit\_float >> (23-m) \;
    Nbit\_binary = 32bit\_binary - pbase\_binary + 1 \;
    \KwIn{Nbit\_to\_32bit(Nbit\_binary)}
    32bit\_binary = Nbit\_binary + pbase\_binary - 1 \;
    32bit\_float = 32bit\_binary << (23-m) \;
    \caption{Offset-based $N$-bit floating point\label{offset_based}}
\end{algorithm2e}

Alg.~\ref{offset_based} summarizes the conversion from 32-bit IEEE 754 to our N-bit offset based float, and N is set w.r.t the precision requirement for the training. Figure~\ref{8bit_example} provides a step-by-step conversion between IEEE 754 and our 8 bits representation.

\begin{figure}[!t]
    \centering
    \subfloat[][(-0.5, 0.5)]{\includegraphics[height=2.5cm]{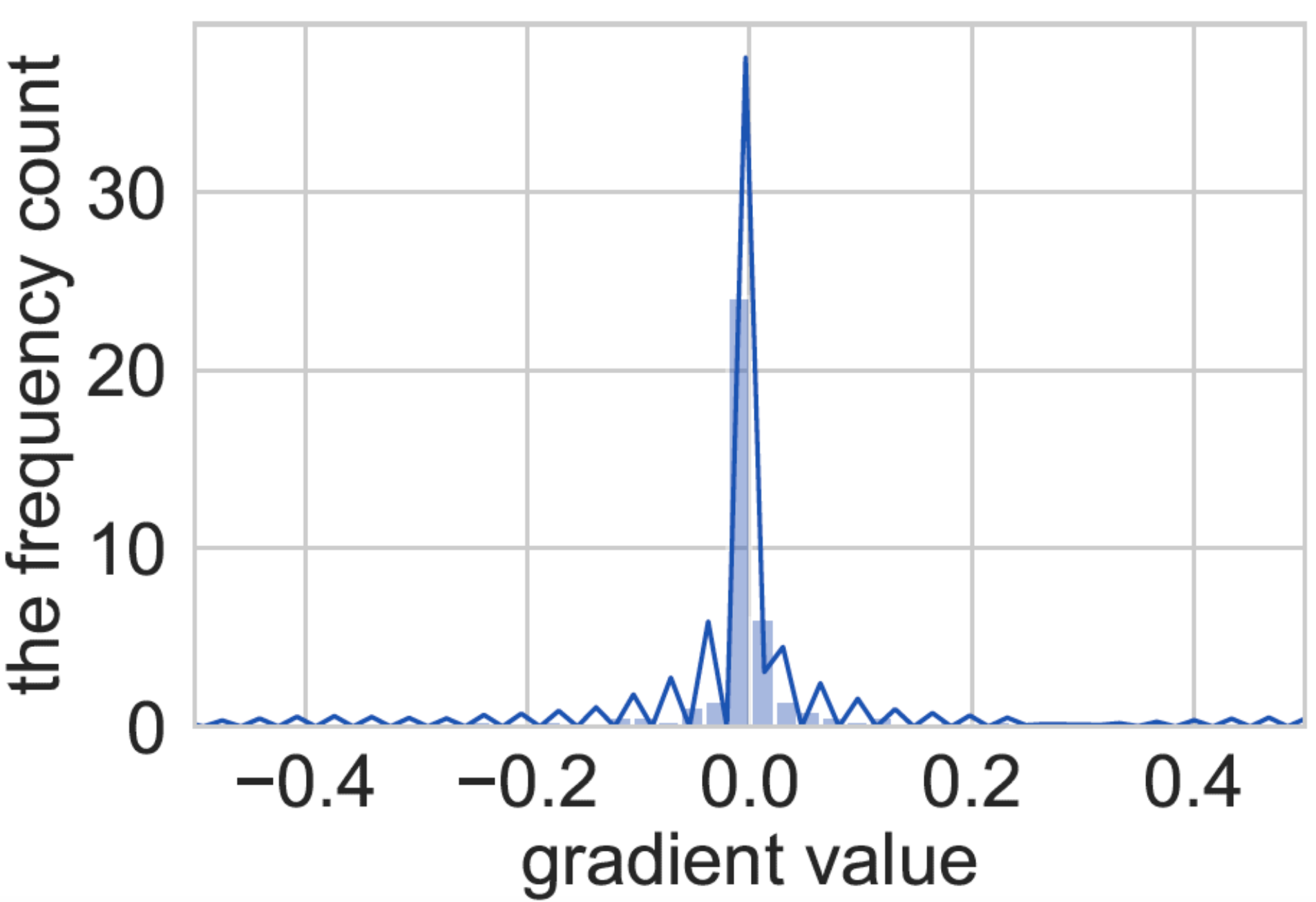}} \quad
    \subfloat[(-5, 5)]{\includegraphics[height=2.5cm]{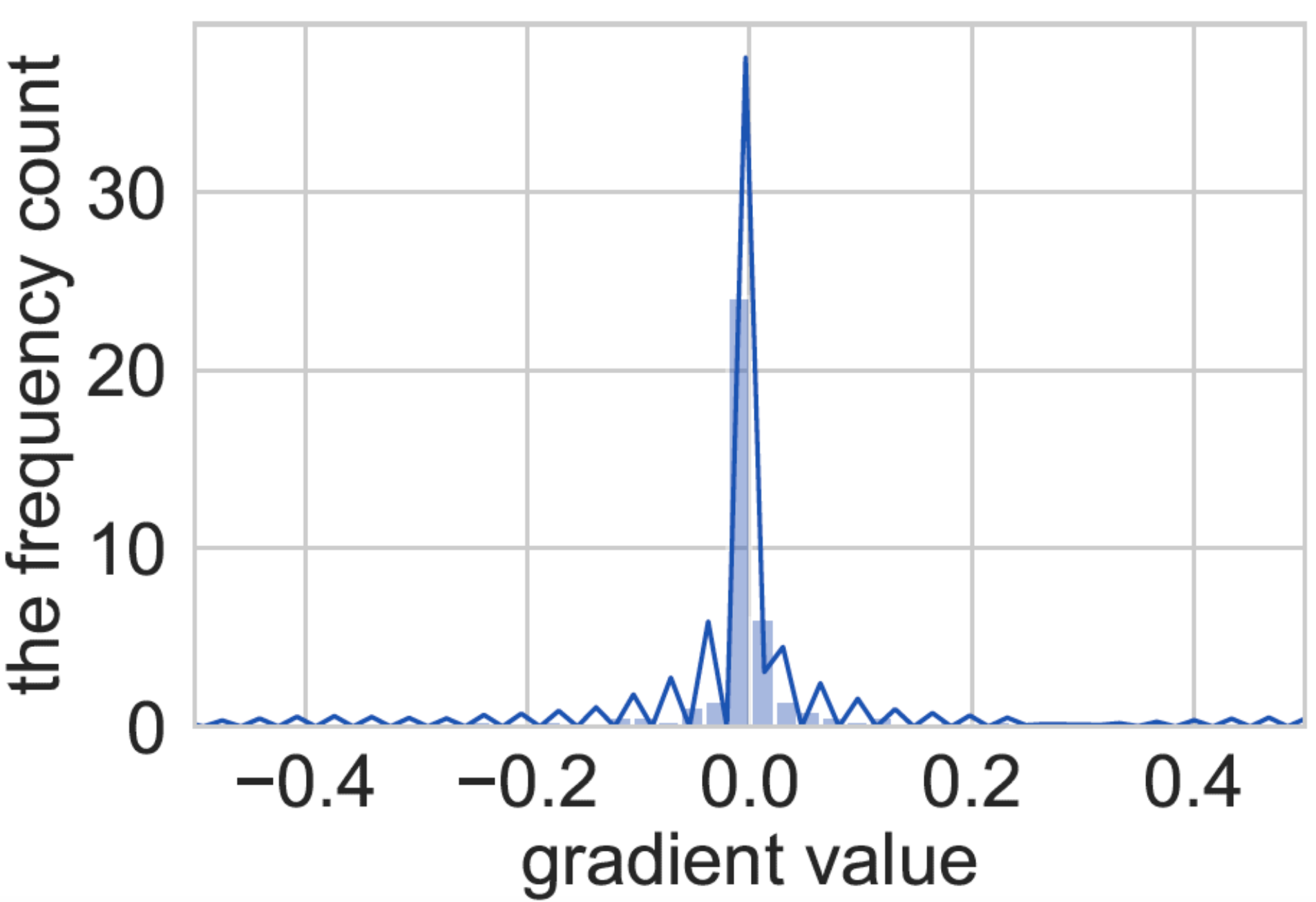}}
\caption{\textbf{Adjustable representation range}: our quantization successfully adjusts its distribution.}
\label{adjustable_representation_range}
\vspace{-0.7cm}
\end{figure}

Figure~\ref{adjustable_representation_range} shows the resulting number distributions of our approach when the range is set to [-0.5, 0.5], and [-5, 5]. This shows our approach successfully adjusts representation ranges, while still maintaining similar distribution to actual gradients. This is because $diff$ increases 2x after $2^m$ numbers, leading to more numbers around 0, and less to $max$ or $min$. Unlike prior static approach, our offset based float dynamically changes the representable range to sustain the various precision requirements from different training tasks. Besides, the float quantizations are embarrassingly data-parallel, so it is easy to achieve the high-performance.

\begin{figure}[!t]
    \centering
    \subfloat[][Top-k selection\label{select}]{\includegraphics[width=0.45\linewidth]{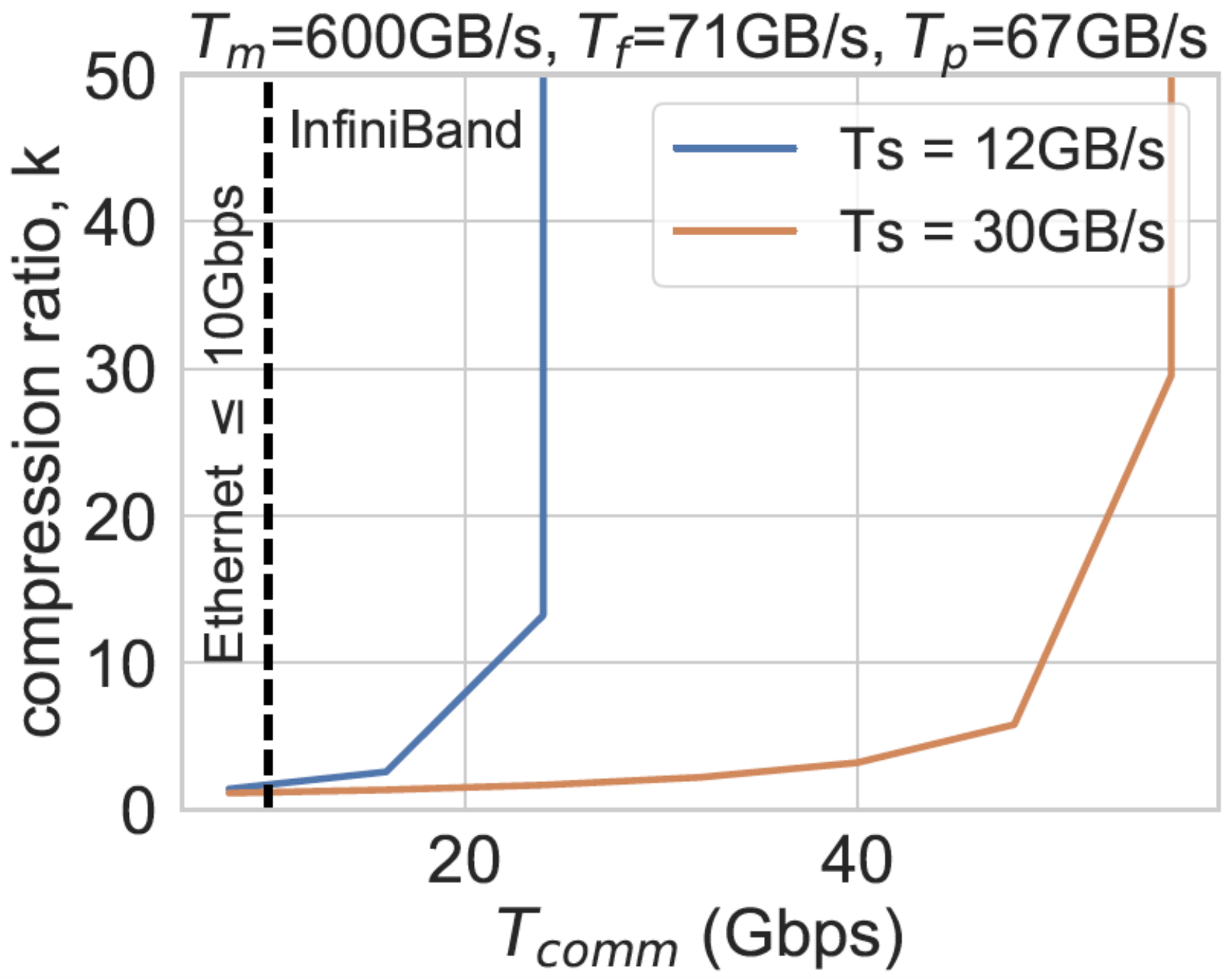}} \quad
	\subfloat[][Packing\label{packing}]{\includegraphics[width=0.45\linewidth]{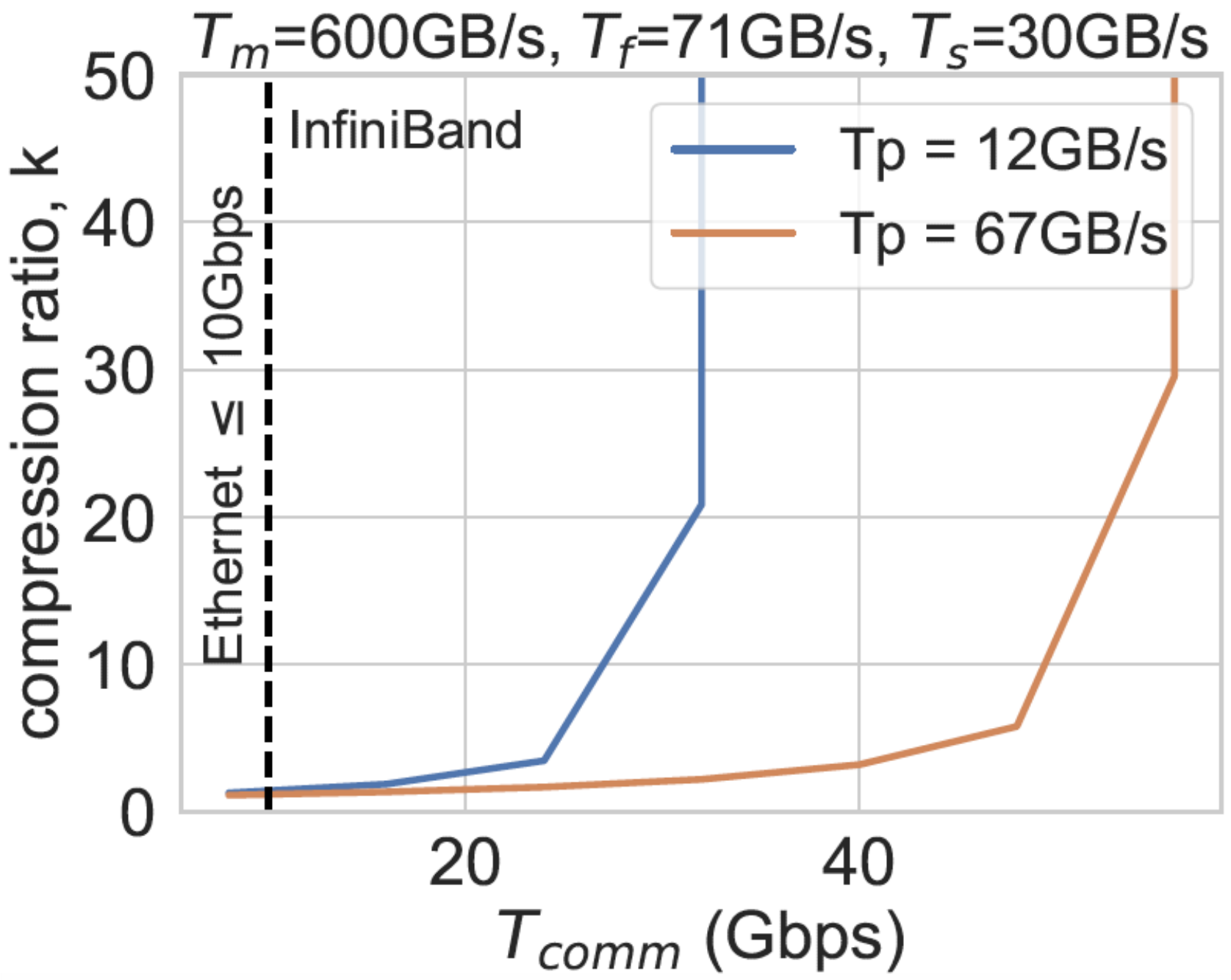}}
    \caption{\textbf{Minimal compression ratio k exhibits performance benefits} at different network bandwidths $T_{comm}$, packing throughput $T_p$ and selection throughput $T_s$. It is easy to get performance improvement from a slow network, while it requires faster compression primitives to be beneficial on a fast network. }
    \label{primtives_comp}
\end{figure}

\subsection{Sensitivity Analysis}
\label{sensitivity_analysis}

\begin{table}[t]
\vspace{0.5cm}
\small
\begin{tabular}{l | l}
\hline
Symbol & Explanation\\ \hline
$T_{m}$ & Maximum throughput of precision conversion including \\
        & float-to-half and range-based quantization\\ \hline
$T_{f}$ & Maximum throughput of FFT \\ \hline
$T_{p}$ & Maximum throughput of packing\\ \hline
$T_{s}$ & Maximum throughput of top-k selection\\ \hline
$T_{comm}$ & Maximum throughput of communication via networks\\ \hline
$k$ & Overall compression ratio \\ \hline
\end{tabular}
\caption{\revise{Symbols of equations in Section~\ref{sensitivity_analysis}.}\label{tab:cost}}
\end{table}

The compression cost shall not offset the compression benefit to
acquire practical performance gain. In this section, we analyze
the performance of compression primitives and their impact on
perceived network bandwidth.
\revise{
Table~\ref{tab:cost} defines all symbols used in the analysis.
It is noted that we use the same notation $T_{m}$ for both float-to-half and
range-based quantization as they are O(N)
algorithms and embarrassingly parallel.
}
Given a message of size $M$, the cost of compression is:
\be
cost_{comp} = M(\frac{2}{T_m}+\frac{1}{T_f}+\frac{1}{T_p}+\frac{1}{T_s})
\label{comp_cost_comp}
\ee
The communication cost after compression is :
\be
cost_{comm} = \frac{M}{T_{comm}}(\frac{1}{k})
\label{saved_comm_cost}
\ee
So the communication cost saved by compression is:
\be
saved\_cost_{comm} = \frac{M}{T_{comm}}(1-\frac{1}{k})
\ee
To compensate for the cost of compression and decompression, \linebreak
$
2cost_{comp} < saved\_cost_{comm}
\label{beneficial_inequality}
$ must hold to acquire the practical performance gain, that is
\be
k > \frac{1}{1 - 2T_{comm}(\frac{2}{T_m}+\frac{1}{T_f}+\frac{1}{T_p}+\frac{1}{T_s})}
\label{k_ratio}
\ee
The performance of $T_m$ depends on the hardware characteristics (such as GPU DRAM bandwidth), and $T_f$ depends on cuFFT.
\revise{It is therefore reasonable to consider them fixed for a particular GPU hardware. 
$T_s$ and $T_p$ depend on the libraries and algorithms applied.
By varying $T_s$ and $T_p$ in Equation~\ref{k_ratio} we analyze the minimal compression ratio $k$ that will show benefits for a particular network infrastructure. 
Figure~\ref{primtives_comp} shows the relationship between $k$ and $T_{comm}$.
If the network throughput is low, like Ethernet, a small $k$ could compensate for the cost
of compression and decompression, which means increasing $k$ would significantly boost
the performance of communications. For example, Figure~\ref{primtives_comp} shows that
$k=2$ is enough to compensate for the overhead of compression and decompression on a 10Gbps
Ethernet.
One the other hand, if the network throughput is high, like InfiniBand, a larger
$k$ would be necessary; otherwise, the overall performance will be impacted by the
overhead of compression and decompression.  More precisely, the red line in
Figure~\ref{primtives_comp} indicates that the minimal compression ratio $k$
should be about 30 to exhibit any benefit on a 56Gbps InfiniBand.

Figure~\ref{primtives_comp} also predicts that the performance of the compression primitives is crucial for high bandwidth networks. As seen in Figure~\ref{primtives_comp}.a, if $T_s$ is 12GB/s, for any $T_comp$ larger than $22 Gbps$ no compression ratio will be able to provide any tangible communication improvement.}

\subsection{Convergence Analysis}
\label{convergence_analysis}
In order to analyse the convergence of our proposed technique we formulate the DNN training as:
\be
\min_x f(x):=\frac{1}{N}\sum_{i=1}^Nf_i(x),
\ee
where $f_i$ is the loss of one data sample to a network. For non-convex optimization, it is sufficient to prove the convergence by showing $ \| \nabla f(x^t) \|^2 \leq \epsilon$ as $t \rightarrow \infty$, where $\epsilon$ is a small constant and $t$ is the iteration. The condition indicates the function converges to the neighborhood of a stationary point. Before stating the theorem, we need to introduce the notion of Lipschitz continuity. $f(x)$ is smooth and non-convex, and $\nabla f$ are $L$-Lipschitz continuous. Namely, $$\|\nabla f(x) - \nabla f(y)\|\leq L\|x-y\|.$$
 For any $x,y$,
 $$f(y)\leq f(x) + \langle \nabla f(x),y-x\rangle + \frac{L}{2}\|x-y\|^2.$$

 \begin{assumption}
 	\label{assumption:variance}
 	Suppose $j$ is a uniform random sample from $\{1,...,N\}$, then we make the following bounded variance assumption:
 	$$\E[\|\nabla f_j(x) - \nabla f(x)\|^2]\leq \sigma^2, \mbox{ for any } x.$$
 \end{assumption}
 This is a standard assumption widely adopted in the SGD convergence proof \cite{nemirovski2009robust} \cite{ghadimi2013stochastic}. It holds if the gradient is bounded.

 \begin{assumption}
 	\label{assumption:Truncation-err}
In the data-parallel training, the gradient of each iteration is $\bar{v}=\frac{1}{p}\sum^{p}_{1}{v_i}$; $p$ is the number of processes, and $v_i$ is the gradient from the $i^{th}$ process. Let's denote $\theta \in [0,1]$ to control the percentage of information loss in the compression function $\hat{v_{i}}=T(v_i, \theta)$ that does quant(FFT-sparsification($v_i$)), so $\bar{\hat{v}} = \sum_{1}^{p}{\hat{v_i}}$. We assume there exists a $\alpha$ such that:
$$\| \bar{v} - \bar{\hat{v}}\|\leq \alpha\|\bar{v}\|.$$
\end{assumption}
So, $\hat{v}$ only loses a small amount of information with respect to $\bar{v}$, and the update from the sparsified gradient is within a bounded error range of true gradient update. It is a necessary condition for deriving the upper bound.

With our compression techniques, one SGD update becomes:
\be
\label{defn:iterate}
\xte = \xt - \eta_t(\frac{1}{P}\sum_{1}^{p}\hat{v}_{i}) = \xt - \eta_t \bar{\hat{v}}_t.
\ee
Then, we have the following lemma for one step:
\begin{lemma}
	\label{lemma:descent}
	Assume $\eta_t\leq\frac{1}{4L}, \theta_t^2\leq \frac{1}{4}$. Then
	\be
	\label{lm:descent}
	\frac{\eta_t}{4}\E[\|\nabla f(\xt)\|^2] \leq \E[f(\xt)] - \E[f(\xte)] + (L\eta_t + \theta_t^2)\frac{\eta_t\sigma^2}{2b_t}.
	\ee
\end{lemma}
Please check the supplemental material for the proof of this lemma. Summing over \eqref{lm:descent} for $K$ iterations, we get:
\bea
\scriptstyle
\label{key-1}
\sum_{t=0}^{K-1}\eta_t\E[\|\nabla f(\xt)\|^2]\leq 4(f(x^0) - f(x^K)) + \sum_{t = 1}^{K-1}(L\eta_t + \theta_t^2)\frac{2\eta_t\sigma^2}{b_t}.
\eea
Next, we present the convergence theorem.
\begin{theorem}
	\label{fixed_lr_theorem}
	\label{theorem:fix-learning-rate} If we choose a fixed learning rate, $\eta_t = \eta$; a fixed dropout ratio in the sparsification function, $\theta_t = \theta$; and a fixed mini-batch size, $b_t= b$; then the following holds:
	$$\scriptstyle \min_{0\leq t\leq K-1}\E[\|\nabla f(\xt)\|^2] \leq \frac{4(f(x^0) - f(x^{K-1}))}{K} + (L\eta + \theta^2)\frac{2\eta\sigma^2}{b}.$$
\end{theorem}
\begin{proof}
$\min_{0\leq t\leq K-1}\E[\|\nabla f(\xt)\|^2] \leq \frac{1}{K} \sum_{t=0}^{K-1}\eta_t\E[\|\nabla f(\xt)\|^2]$, as $\|\nabla f(\xt)\|^2 \geq 0$. By \eqref{key-1}, we get the theorem.
\end{proof}

\begin{theorem}
\label{diminishing_lr_theorem}
If we apply the diminishing stepsize, $\eta_t$, satisfying $\sum_{t=0}^{\infty}\eta_t = \infty, \sum_{t = 0}^{\infty}\eta_t^2 <\infty$,
our compression algorithm guarantees convergence with a diminishing drop-out ratio, $\theta_t$, if $\theta_t^2 = L\eta_t$.
\end{theorem}

\begin{proof}
If we randomly choose the output, $x_{out}$, from $\{x^0,...,x^{K-1}\}$, with probability $\frac{\eta_t}{\sum_{t=0}^{K-1}\eta_t}$ for $\xt$, then we have:
\bea
\E[\|\nabla f(x_{out})\|^2]  & = \frac{\sum_{t=0}^{K-1}\eta_t\E[\|\nabla f(\xt)\|^2]}{\sum_{t=0}^{K-1}\eta_t} \\
                             & \leq \frac{4(f(x^0)-f(x^*)}{\sum_{t=0}^{K-1}\eta_t} + \frac{\sum_{t=0}^{K-1}(L\eta_t + \theta_t^2)2\eta_t\sigma^2}{b\sum_{t=0}^{K-1}\eta_t}.
\eea

Note that $\sum_{t=0}^{K-1}\eta_t\rightarrow\infty$, while\\
$\sum_{t=0}^{K-1}(L\eta_t + \theta_t^2)2\eta_t\sigma^2 = \sum_{t=0}^{K-1}4L\eta_t^2\sigma^2<\infty$, \\
and we have $\E [\|\nabla f(x_{out})\|^2]\rightarrow0$.
\end{proof}

\section{Evaluation} \label{evaluation}
Our experiments consist of two parts to assess the proposed techniques. First, we validate the convergence theory and its assumptions with AlexNet on ImageNet and ResNet32 on CIFAR10, which sufficiently cover typical workloads in traditional linear and recent non-linear neural architectures, \revise{and also provide coverage on two widely used datasets.}
Then, we show that the FFT-based method demonstrates better convergence and faster compression than other state-of-the-art compression methods such as QSGD \cite{alistarh2017qsgd}, TernGrad \cite{wen2017terngrad}, Top-k sparsification \cite{lin2017deep, alistarh2018convergence}, as our techniques incur fewer approximation errors, while still delivering a competitive compression ratio for using both sparsification and quantization.

%

\textbf{Parallelization scheme}: we choose BSP for parallelization for its simplicity in the theoretical analysis:
BSP follows strict synchronizations, allowing us to better observe the effects of gradient compression toward the convergence by iterations.

\textbf{Implementation}: we implemented our approach, \revise{losses} SGD(\revise{no compression}), QSGD, Top-K, and TernGrad in a C++ DL framework, SuperNeurons \cite{Wang:2018:SDG:3178487.3178491}; We used the allgather collective from NVIDIA NCCL2
to exchange compressed gradients since existing communication libraries lack the support for sparse all-reduce (Figure~\ref{mpi_integration}). 
\revise{Even though SGD usually uses allreduce instead of allgather as it does not have compression; for a fair comparison, we applied allgather for all algorithms
to demonstrate the algorithmic benefit of our FFT compression.}
Every GPU has a copy of global gradients for updating parameters after all-gather local gradients. Parameters need to be synchronized after multiple iterations to eliminate the precision errors, and here we broadcast parameters every 10 iterations.
\revise{It is noticed that we did not adopt communication and computation overlapping strategy as it could be another optimization method orthogonal to compression, and is not in
the scope of this paper.}
\begin{figure}[!t]
    \centering
	\hspace{-0.7cm}
    \subfloat[][AlexNet]{\includegraphics[height=2.8cm]{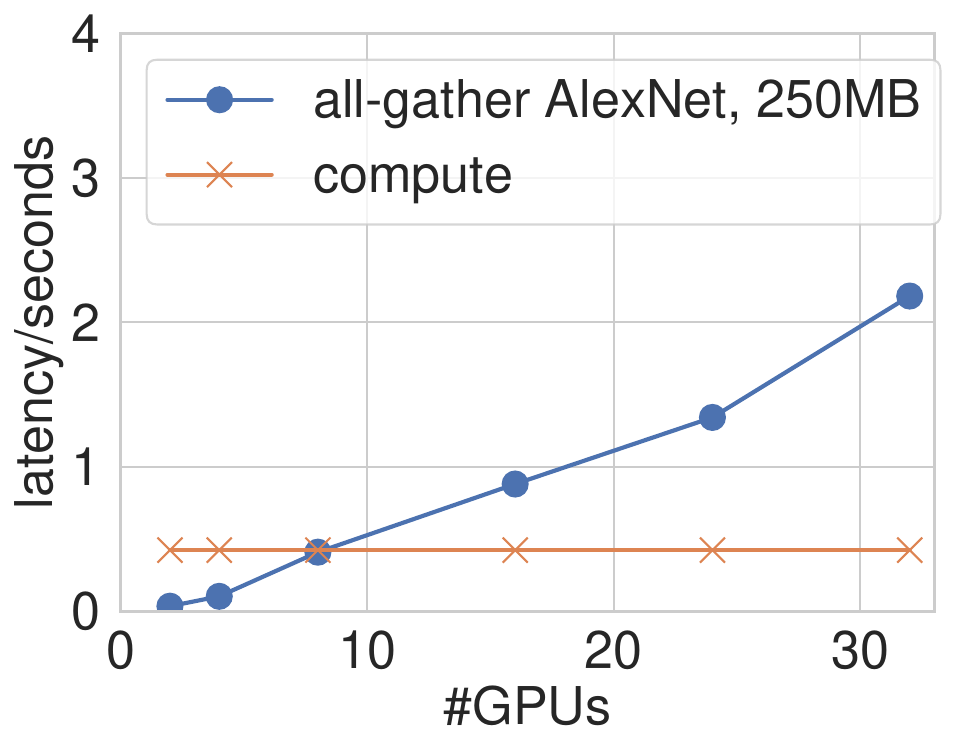}}  \quad
    \subfloat[][ResNet32]{\includegraphics[height=2.8cm]{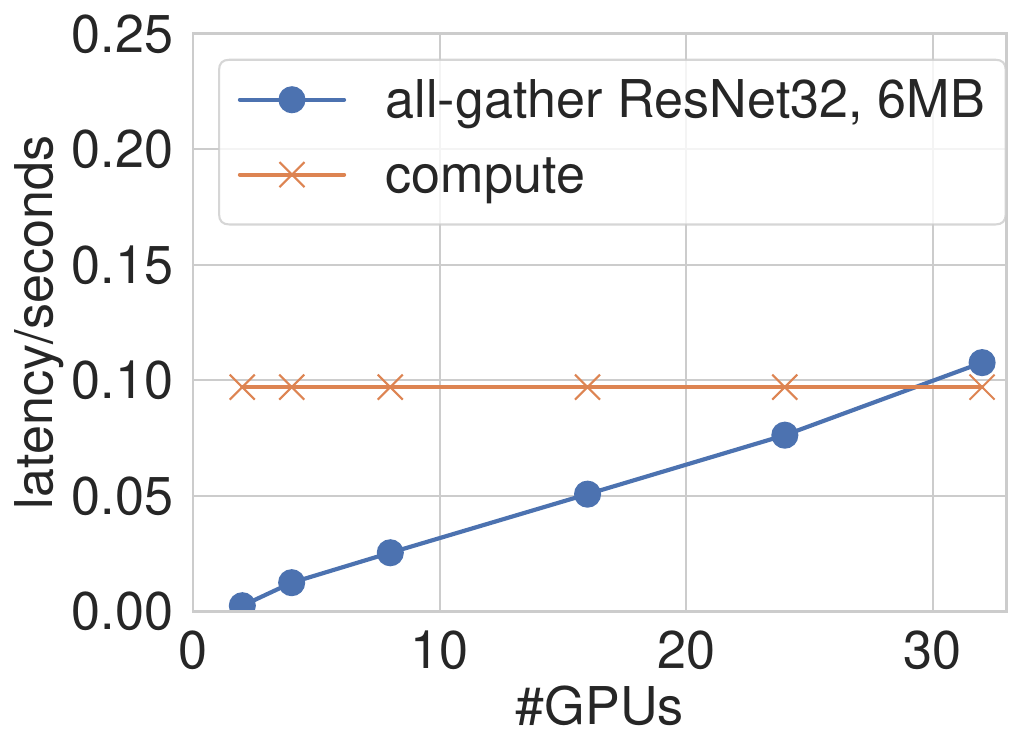}}
    \caption{\textbf{the latency for \textit{all-gather}} AlexNet and ResNet32 from 2 to 32 GPUs.  }
    \label{all_gather_speed}
\end{figure}

\textbf{Training setup}: The single GPU batch is set to 128 and 64 for ResNet32 and AlexNet, respectively. The momentum for both networks is set to 0.9. The learning rate for Resnet32 is 0.01 at epochs $\in[0, 130]$, and 0.001 afterwards; the learning rate for AlexNet is 0.01 at epochs $\in [0, 30]$, 0.001 at epochs $\in [30, 60]$, and 0.0001 afterwards.   

\textbf{Machine setup}: we conducted experiments on the Comet cluster hosted at San Diego Supercomputer Center. Comet has 36 GPU nodes, and each node is equipped with 4 NVIDIA TESLA P100 GPUs and 56 Gbps FDR InfiniBand, Figure~\ref{all_gather_speed} shows the \textit{allgather} cost almost linearly increases with the number of GPUs. This is because the total exchanged messages in \textit{allgather} linearly increase with \#GPUs~\cite{gabriel2004open}.
In our experiments, we used 8 GPUs in evaluating the accuracy and performance when integrating our compression methods in training, and up to 32 GPUs in evaluating the scalability
of the distributed training.

\subsection{Validation of Theorems}

\begin{figure}[!t]
    \centering
    \subfloat[][AlexNet]{\includegraphics[width=0.65\linewidth]{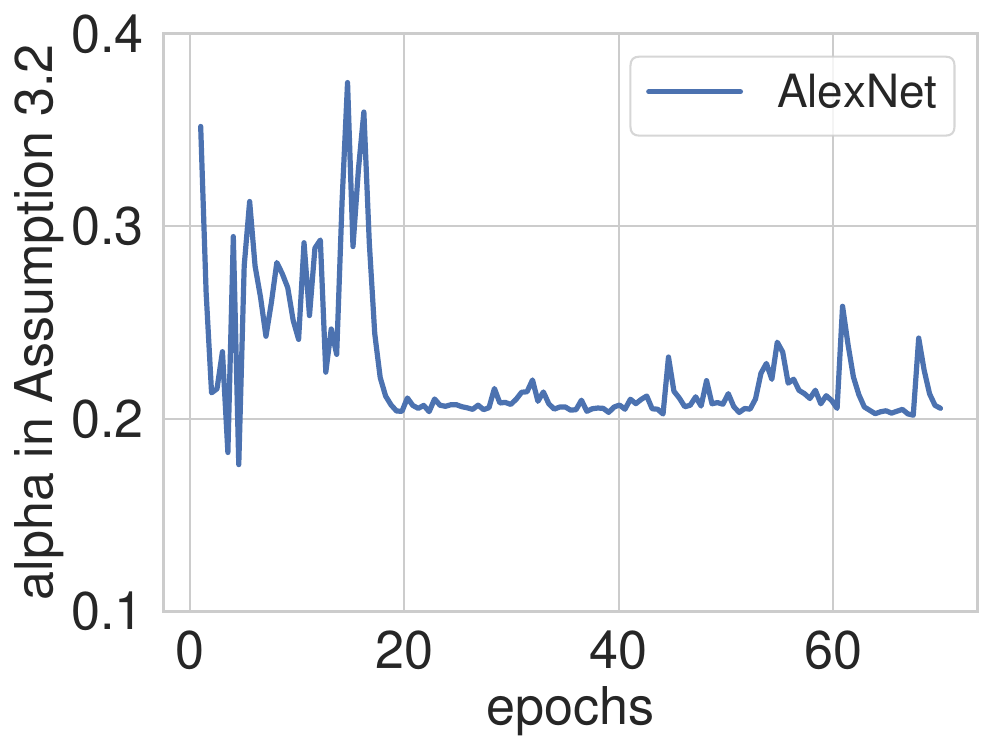}} \quad
    \subfloat[][ResNet32]{\includegraphics[width=0.65\linewidth]{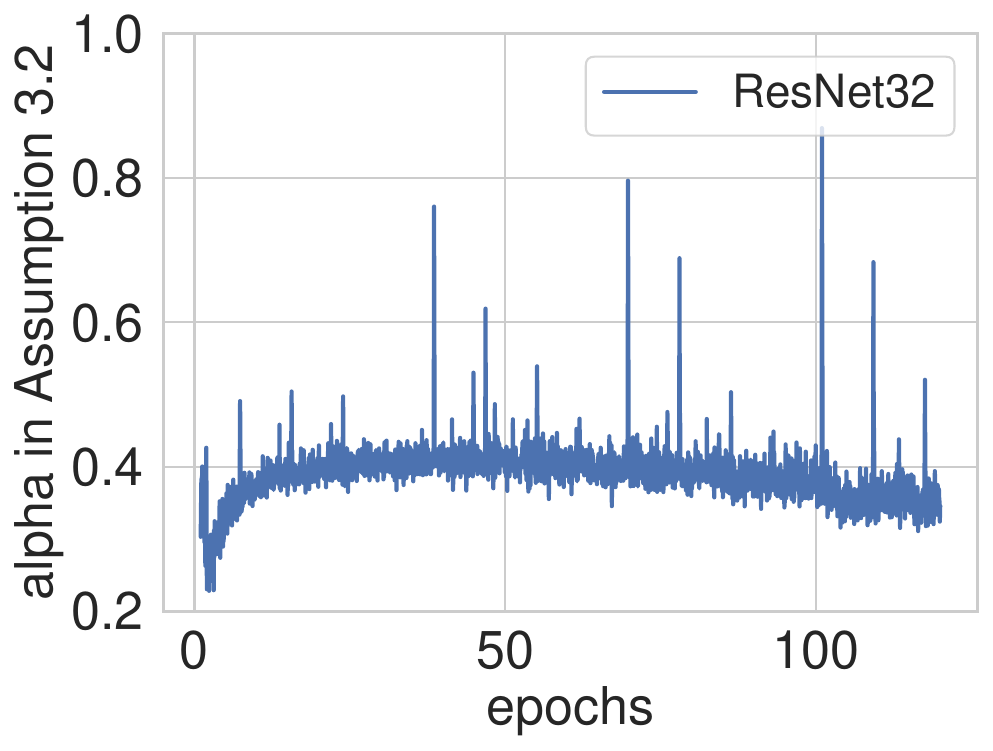}}	
    \caption{Empirical verification of Assumption~\ref{assumption:Truncation-err}.}
    \label{empirical_validation_assumption}
    \vspace{-0.6cm}
\end{figure}

\begin{figure}[!t]
    \centering
    \subfloat[][AlexNet]{\includegraphics[width=0.65\linewidth]{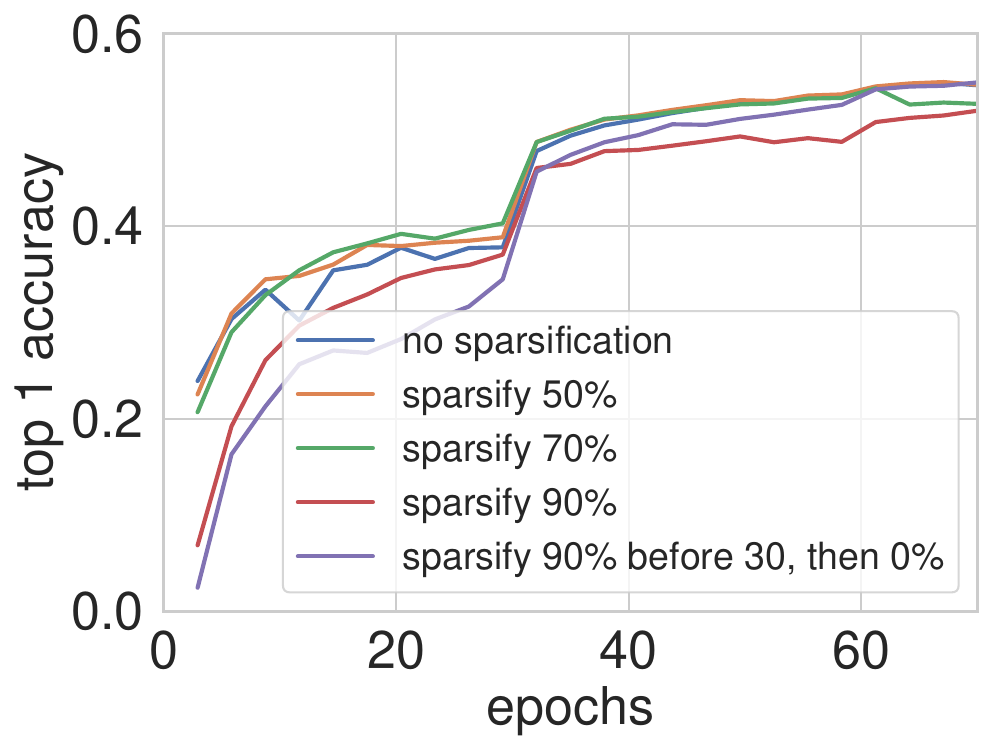}} \quad
    \subfloat[][ResNet32]{\includegraphics[width=0.65\linewidth]{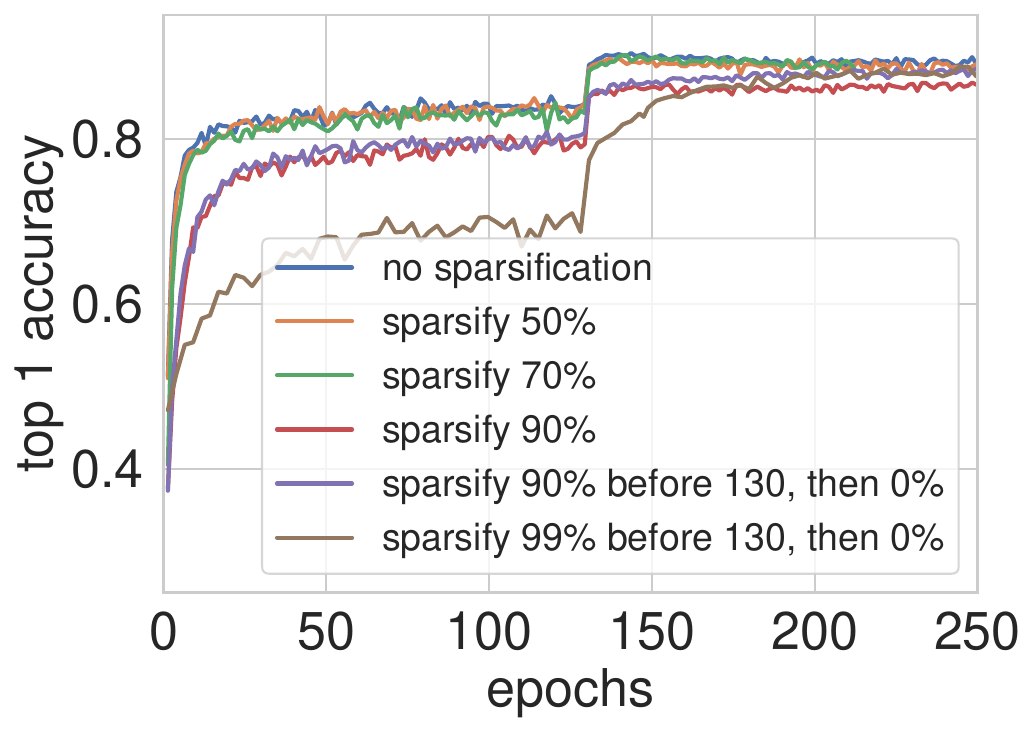}}
    \caption{Empirical validation of Theorem~\ref{diminishing_lr_theorem}.}
    \label{empirical_validation}
    \vspace{-0.5cm}
\end{figure}


\textbf{Verification of assumptions}: our convergence theorems rely on Assumption~\ref{assumption:variance} and Assumption~\ref{assumption:Truncation-err}. Assumption~\ref{assumption:variance} automatically holds due to the bounded gradients. Assumption~\ref{assumption:Truncation-err} always hold if $p = 1$, but it can break in very rare cases for $p > 1$. For example, $\alpha$ does not exist if $\bar{v} = [0, 0]$, given two opposite gradients, e.g. $\bar{v}_{1} = [-0.3, 0.5]$ and $\bar{v}_{2} = [0.3, -0.5]$. Though the scenario is very unlikely, we empirically validate Assumption~\ref{assumption:Truncation-err} on different training tasks by calculating $\alpha = \frac{\| \bar{v} - \bar{\hat{v}}\|}{\|\bar{v}\|}$. From Figure~\ref{empirical_validation_assumption}, $\alpha \in [0, 1]$ practically sustaining Assumption~\ref{assumption:Truncation-err}.

\textbf{Validation of theorems}: Theorem~\ref{fixed_lr_theorem} states a large compression ratio, i.e. large $\theta$, can jeopardize the convergence, and theorem~\ref{diminishing_lr_theorem} states that our FFT-based sparsified SGD is guaranteed to converge with a diminishing compression ratio. The goal of optimization is to find a local optimum, where the gradient approximates to zero, i.e, $\E[\|\nabla f(\xt)\|^2] \rightarrow 0$, as $K\rightarrow\infty$. From the inequality in theorem~\ref{fixed_lr_theorem}, $\frac{4(f(x^0) - f(x^{K-1}))}{K} \rightarrow 0$ as $K\rightarrow\infty$, leaving $E[\|\nabla f(\xt)\|^2]$ bounded by $(L\eta + \theta^2)\frac{2\eta\sigma^2}{b}$. $L\eta\frac{2\eta\sigma^2}{b}$ is the error term from SGD, and $\theta^2\frac{2\eta\sigma^2}{b}$ is the error term from the compression. Compared to the SGD, using a large $\theta$ in the gradient compression slacks off the bound for $\E[\|\nabla f(\xt)\|^2]$, causing the deterioration on both the validation accuracy and training loss. As shown in Figure~\ref{empirical_validation}, when $\theta = 0.5$ (i.e., sparsify 50\%), the accuracy and loss traces of AlexNet and ResNet32 behave exactly the same as SGD (\revise{shown as no sparsification}). When $\theta = 0.9$ (i.e., sparsify 90\%), both the training loss and validation accuracy significantly deviate from SGD, as a large $\theta$ increases the error term $\frac{2\eta\sigma^2\theta^2}{b}$ loosening the bound for $\E[\|\nabla f(\xt)\|^2]$. Therefore, $\theta$ is critical to retain the same performance as SGD, and it is tricky to select $\theta$ in practice. 
We present Theorem.\ref{diminishing_lr_theorem} to resolve this issue. The theorem compensates for Theorem~\ref{fixed_lr_theorem}, indicating that a large $\theta$ can still deliver the same accuracy as SGD if we shrink the $\theta$ during the training. Empirical results in Figure~\ref{empirical_validation} validate Theorem~\ref{diminishing_lr_theorem}. For example, by setting $\theta = 0.9$ (drop 90\%, \revise{red line}), both AlexNet and ResNet32 fail to converge to the same case of SGD. However, it is able to bring the accuracy back to the same result as the SGD in the same epochs simply by diminishing $\theta$ from 0.9 to 0 at the 30th epoch for AlexNet, and at the 130th epoch for ResNet32. Therefore, we claim both Theorem~\ref{fixed_lr_theorem} and Theorem~\ref{diminishing_lr_theorem} are legitimate.

\textbf{Implications of theorems}:
\revise{these two theorems explain the relationship between the accuracy and compression ratio $\theta$, and act as a guide to help preserve the training network accuracy by tuning the compression ratio during the training.}
Hence, in practice, to ensure the convergence, we can shrink $\theta$ along with the learning rate $\eta$ for the condition of $\theta_t^2 = L \eta_t$. In order to recover the accuracy, we can also reduce $\theta$ as the case in Fig.~\ref{empirical_validation} that a failure case ($\theta=0.9$) recovers the accuracy after reducing $\theta$ to 0 in the middle of training.

\subsection{Algorithm Comparisons}

\begin{figure}[!t]
    \centering
    \subfloat[][AlexNet]{\includegraphics[width=0.65\linewidth]{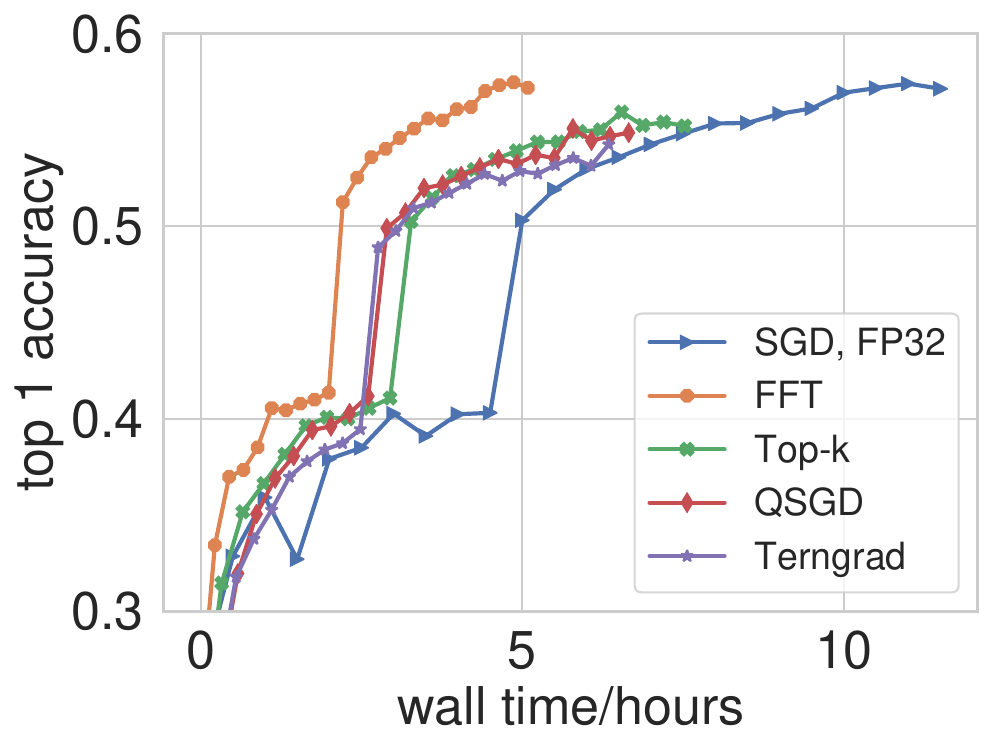}}
    
    \subfloat[][ResNet32]{\includegraphics[width=0.65\linewidth]{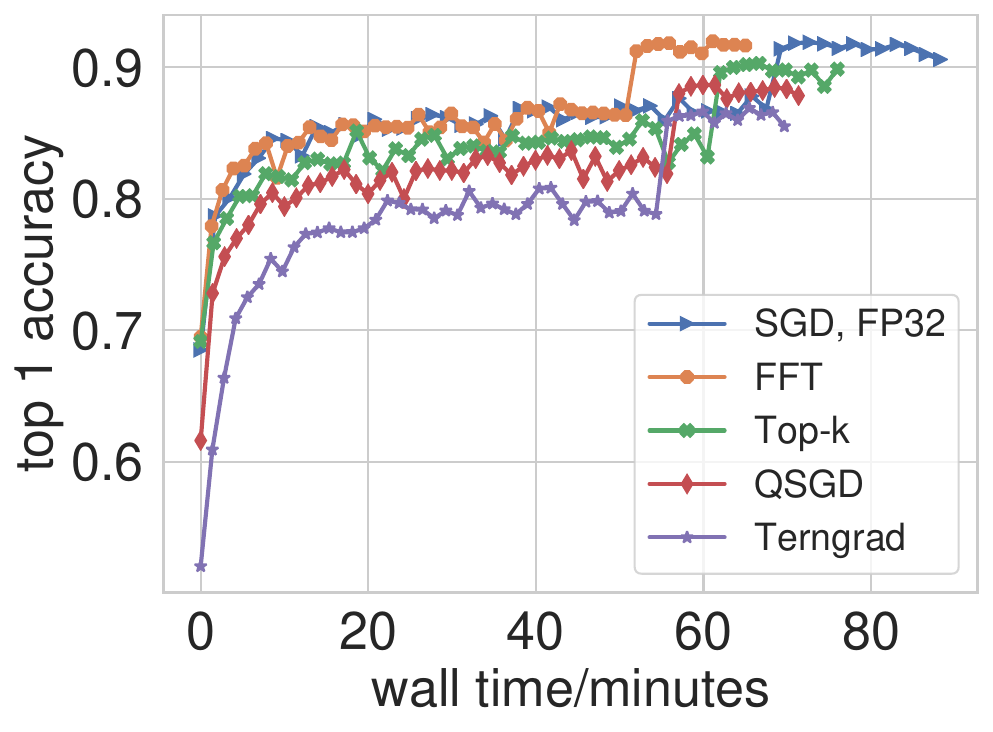}}
    \caption{\textbf{Training wall time on a 8 GPUs cluster}: FFT outperforms TernGrad, QSGD and Top-k in both the speed and test accuracy. FFT is faster for a high compression ratio by combining sparsification and quantization, while the better gradient quality of FFT explains the good accuracy, as we will show in Figure~\ref{gradient_dist_alg_comps}. }
    \label{training_walltime}
\end{figure}

\begin{table}
\small
\setlength{\tabcolsep}{0.15em}
\begin{tabular}{c c c c c c c c}
  \hline
    Method                &   \begin{tabular}{@{}c@{}}AlexNet \\ top1 acc\end{tabular}  & \begin{tabular}{@{}c@{}}Speedup \\ w.r.t SGD\end{tabular}  &  \begin{tabular}{@{}c@{}}ResNet32 \\ top1 acc\end{tabular}      & \begin{tabular}{@{}c@{}}Speedup \\ w.r.t SGD\end{tabular}  \\   \hline
    SGD, FP32                   &   56.52\%                         &1    &   92.11\%                 &  1       \\
    FFT                   &   56.61\%, ($+0.09\%$)     &2.26   &  91.99\%, ($-0.12\%$)          &  1.33x    \\
    Top-K                 &   55.07\%, ($-1.45\%$)     &1.53   &  90.31\%, ($-1.80\%$)          &  1.12x   \\
    QSGD                  &   53.54\%, ($-2.98\%$)     &1.73   &  88.66\%, ($-3.45\%$)          &  1.21x   \\
    TernGrad-noclip       &   52.86\%, ($-3.66\%$)     &1.81   &  86.90\%, ($-5.21\%$)          &  1.24x   \\
   \hline
\end{tabular}
\vspace{0.2cm}
\caption{\textbf{Summarization of Figure~\ref{training_walltime}}: the difference of test accuracy and the speedup over lossless SGD.}
\label{algorithm_results}
\vspace{-0.4cm}
\end{table}

\textbf{\textit{Choice of Algorithms}}: Here we evaluate our FFT-based techniques against 3 major gradient compression algorithms, Top-k sparsification~\cite{lin2017deep, alistarh2018convergence, aji2017sparse}, and Terngrad~\cite{wen2017terngrad} and QSGD~\cite{alistarh2017qsgd}. The baseline method is SGD using 32 bits float. Top-k sparsification thresholds the gradients w.r.t their magnitude, and the compression ratio is determined by 1/(1-$\theta$), where $\theta$ is the drop-out ratio. Please note that Top-k variant e.g. DGC~\cite{lin2017deep} utilizes heuristics like error accumulation and momentum correction to boost performance. To fairly evaluate Top-k sparsification against FFT based sparsification, we evaluated the vanilla Top-k v.s. the vanilla FFT sparsification, and finding heuristics to boost FFT sparsification is orthogonal to this study. Both Terngrad and QSGD map gradients to a discrete set. Specifically, Terngrad maps each gradient to the set of $\{-1, 0, 1\}*max(|g|)$, and thus 2 bits are sufficient to encode a gradient. Instead, QSGD uses $N$ bits to maps each gradient to a uniformly distributed discrete set containing $2^N$ bins. Please note TernGrad does not quantize the last classification layer to keep good performance \cite{wen2017terngrad}, while we sparsify the entire gradients.

\textbf{Algorithm Setup}: Regarding Top-k and FFT based sparsification, results from Figure~\ref{empirical_validation} and \cite{alistarh2018convergence} show a noticeable convergence slowdown after $\theta > 90\%$.
To maintain a reasonable accuracy, we choose $\theta$ = 85\% for both top-k and FFT based sparsification. We use $min = -1$ and $max = 1$ as the boundaries, and 10 bits in initializing our N-bit quantizer. Therefore, the compression ratio for Top-k is 1/(1-$\theta$) = 6.67x and FFT based is 21.3x with an additional 32/10 from quantizers. Terngrad uses 2 bits to encode a gradient, while we use 8 bins (3 bits) for QSGD to encode a gradient. As a result, the compression ratio of Terngrad is 16x and QSGD is 10.6x. Please note we calculate the compression ratio w.r.t gradients as gradient exchanges dominate communications in BSP. Following a similar setup in Figure~\ref{empirical_validation}, each algorithm is set to run 180 epochs on CIFAR10 and 70 epochs on ImageNet using 8 GPUs.


Figure~\ref{training_walltime} demonstrates that our framework outperforms QSGD, Terngrad, and Top-k in both the final accuracy and the training wall time on an 8 GPU cluster, and Table~\ref{algorithm_results} summarizes the \revise{test accuracy and speedup over the lossless SGD}. Particularly, FFT consistently reaches a similar accuracy to SGD with the highest speedup. To further investigate the algorithmic and system advantages of the FFT method, we investigate the gradient quality and the scalability of iteration throughput.

\begin{figure*}[!t]
    \centering
    \subfloat[][Ours\label{fft_sparsified_dist}]{\includegraphics[height=2.3cm]{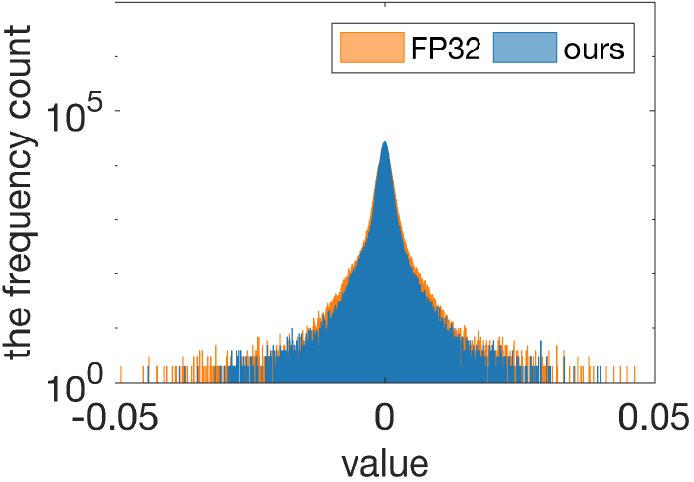}} 
    \subfloat[][Top-k\label{terngrad_dist}]{\includegraphics[height=2.3cm]{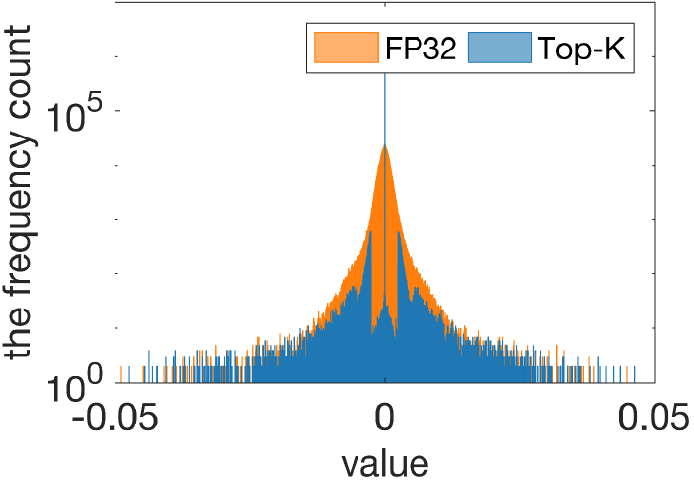}} 	
    \subfloat[][Terngrad\label{direct_sparsified_dist}]{\includegraphics[height=2.3cm]{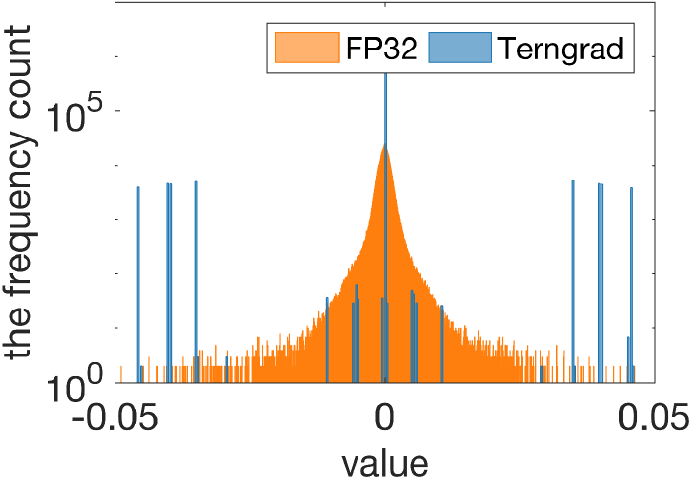}} 
    \subfloat[][QSGD\label{qsgd_dist}]{\includegraphics[height=2.3cm]{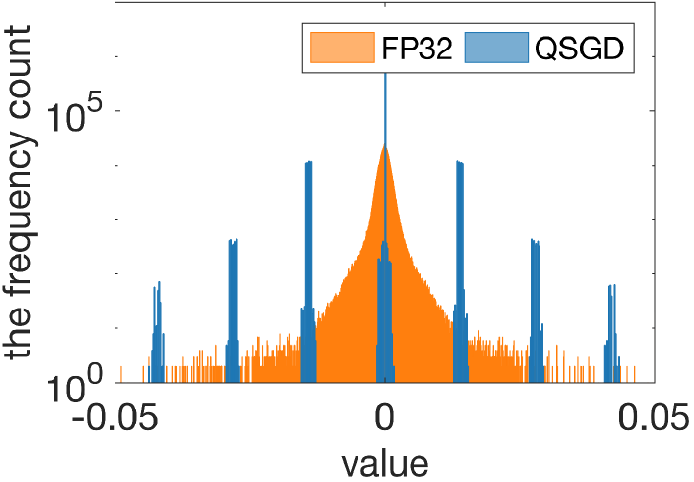}} 
    \subfloat[][reconstruction error\label{recon_error}]{\includegraphics[height=2.3cm]{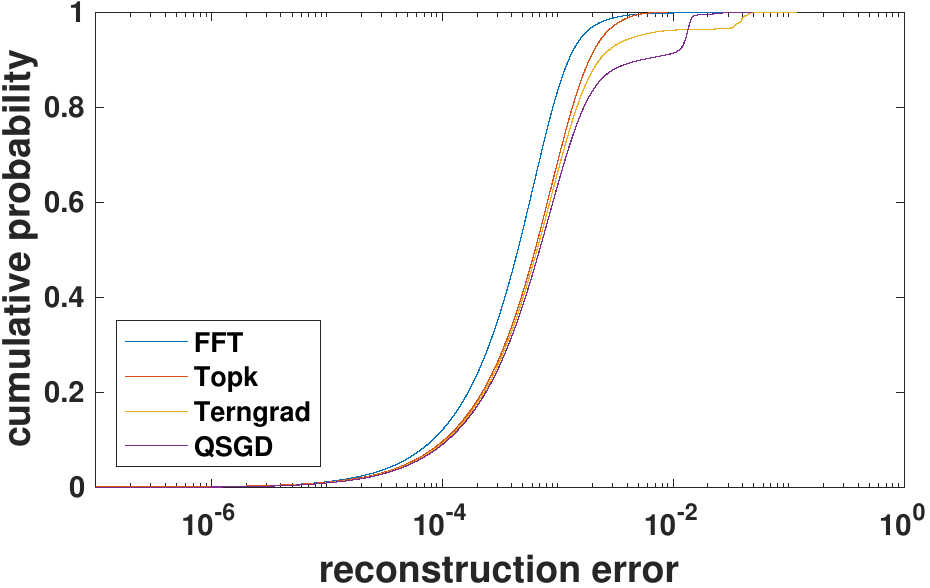}} 
    \caption{ \textbf{(a)$\rightarrow$(d)}: Histogram of reconstructed gradients (blue) by FFT ($\theta = 0.85$), Top-k ($\theta = 0.85$), QSGD and Terngrad v.s. the original. The reconstructed gradients by FFT is the closest to the original(FP32). \textbf{(e)} Cumulative error distribution of $|g_{i} - \hat{g}_{i}|$, where $g_{i}$ is the $i$-th true gradient, and $\hat{g}_{i}$ is the $i$-th sparsified gradient. FFT incurs less errors than other approaches for 99.7\% of the gradients. }
    \label{gradient_dist_alg_comps}
    \vspace{-0.5cm}
\end{figure*}

\subsubsection{The algorithmic advantages of FFT}\label{algorithm_adv}
We claim the algorithmic advantages of FFT for preserving the original gradient distribution and rendering fewer reconstruction errors than others. We uniformly sampled the gradients of ResNet32 every 10 epochs during the training. Figure~\ref{gradient_dist_alg_comps} demonstrates the distribution of reconstructed gradients w.r.t the gradients before the compression. FFT is the only one that retains the original gradient distribution, though $\theta = 85\%$ frequency has been removed. In contrast, Top-k loses the peak for eliminating the near-zero elements at the same $\theta$. Similarly, QSGD presents 7 clusters for using 8 bins to represent a gradient; and, in general, TernGrad shows 3 major clusters around \{0, -0.05, 0.05\} for using a quantization set of \{-1, 0, 1\}. Please note that Terngrad shows 11 bars; this is due to the aggregation of sparsified gradients from each node. Aside from qualitatively inspecting the gradient distribution, we also quantitatively examined the empirical cumulative distribution of the reconstruction error in Figure~\ref{recon_error}. FFT demonstrates the lowest error within the range of $[10^{-5}, 10^{-2}]$. Therefore, FFT can reach better accuracy in the same training iterations.

\begin{figure}[!t]
    \centering
	    \subfloat[][AlexNet\label{alexnet_ps_imp}]{\includegraphics[height=2.6cm]{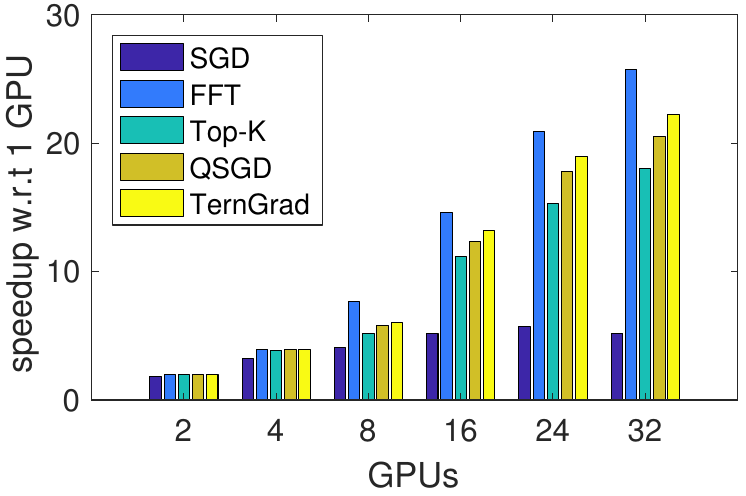}} \quad
	    \subfloat[][ResNet32\label{resnet_ps_imp}]{\includegraphics[height=2.6cm]{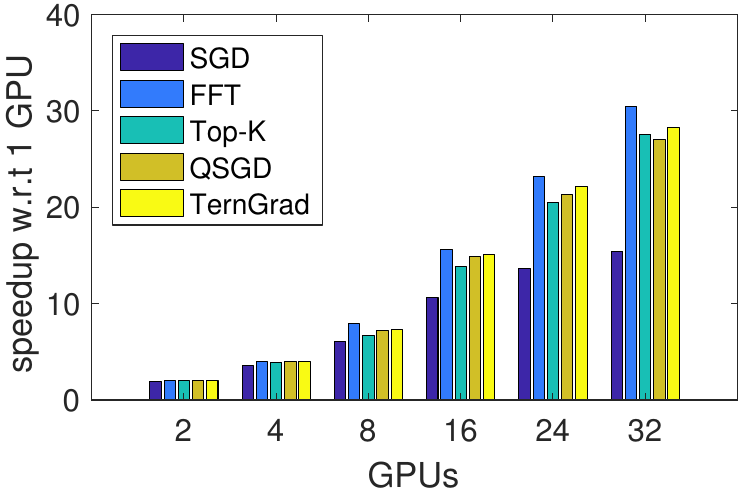}}
    \caption{\textit{\textbf{\revise{Weak} scalability from 2 to 32 GPUs}}: we measure the iteration throughput, and calculate the speedup w.r.t 1 GPU.}
    \label{scalability}
\end{figure}

\subsubsection{The system advantages of FFT}\label{system_adv} 
Our compression framework fully exploits both the gradient sparsity and the redundancy in 32-bit floating point by further quantizing the FFT sparsified gradient. It enables FFT to deliver a much higher iteration throughput than QSGD, TernGrad, and Top-K. Following the same setting in Figures~\ref{training_walltime}, Figure~\ref{scalability} demonstrates the iteration throughput of training AlexNet and ResNet32 from 2 to 32 GPUs. Please note that using a very large $\theta$ (e.g., 0.999) can get an impressive speedup, but it also drastically hurts the final accuracy. Here we still use $\theta=85\%$. The gradients of AlexNet (ImageNet) is around 250 MB, while the gradients of ResNet32 (CIFAR-10) are only 6MB. Therefore, the scalability of AlexNet is generally better than ResNet32. Better results are also observable if using a slow network, e.g., 100MB Gbps. When GPUs $\leq 4$, the speedup is similar as communications are intra-node through PCI-E. FFT still consistently demonstrates the highest iteration throughput for a better compression ratio when GPUs increase from 8 to 32.

\section{Related Work}
\label{related_work}
We categorize the existing lossy gradient compression into two groups: (1) \emph{quantization} and (2) \emph{sparsification}.

\textit{Quantization}: 1-bit SGD~\cite{seide20141} is among the first to quantize gradients to alleviate the communication cost in the distributed training. Specifically, it quantizes a 32-bit IEEE-754 float into a binary of [0, 1] to achieve a compression ratio of 32$\times$. Though their methods are purely heuristic, and their empirical validations demonstrate a slight loss of accuracy, it shows the possibility to train a network with highly lossy gradients. Subsequently, several quantization methods have been proposed. Flexpoint~\cite{flexpoint} uses block floating-point encoding based on current gradient/weight values. HOGWILD!~\cite{de2015taming} quantizes both weights and gradients into 8-bit integers by rounding off floats (i.e., low-precision training); but this idea is largely restricted by the availability of low-precision instruction sets. TernGrad~\cite{wen2017terngrad} quantizes a gradient as [-1, 0, 1]$*|max(g)|$, while QSGD~\cite{alistarh2017qsgd} stochastically quantizes gradients onto a uniformly discretized set. Both approaches distribute the precision uniformly across the representable range---ignoring both the distribution and the range of the gradients. As we show, gradients follow a normal distribution (Figure~\ref{grad_dist}). In our range-based quantizer, we allocate precision for the range and the distribution of the values to better exploit the limited number of bits. Most importantly, QSGD and TernGrad damage the original gradient distribution due to limited representable values after the quantization (Figure~\ref{gradient_dist_alg_comps}). As a result, TernGrad and QSGD incur an observable deterioration in the final accuracy (Table~\ref{algorithm_results}).

\textit{Sparsification}: Aji and Heafield~\cite{aji2017sparse} present the very first Top-k gradient sparsification showing that the training can be done with a small accuracy loss by setting the $99\%$ smallest gradients to zeros. Based on the Top-k thresholding, Han et al.~\cite{han2015deep} propose Deep Compression, which uses heuristics like momentum correction and error accumulation to resolve the accuracy loss in the vanilla Top-k. Please note that these heuristics are orthogonal to our methods and can also be applied to improve ours. Jin et al.~\cite{deepsz} propose DEEPSZ, which performs error-bounded
lossy compression on the pruned weight. It is a modification of the SZ lossy compression framework~\cite{sz}.
Cédric et al.~\cite{sparsecml} propose a communication sparsification approach called SPARCML. Different from ours, the SPARCML focuses on the implementation of MPI collective operations of sparse data.
D. Alistarh et al.~\cite{alistarh2018convergence} analyze the convergence of Top-k compression. With \cite{alistarh2018convergence}, we noticed a significant convergence slowdown at a large sparsity. As we investigated, these Top-k methods also distort the gradient distribution at a large sparsity, yielding higher approximation errors than the original gradients. At the same sparsity ($\theta$), our FFT method is much better at preserving the original gradient distribution and shows less approximation error and better results.

\section{Conclusion}

As indicated in Sec.~\ref{motivation}, exchanging gradients is the major bottleneck for the distributed DNN training. To alleviate this communication bottleneck, this paper proposes a lossy gradient compression framework that uses an FFT-based gradient sparsification and a range-based, variable-precision, floating-point representation. We theoretically prove that our techniques preserve the convergence and the final accuracy by adapting the sparsification ratio $\theta$ during the training, and empirically verify the assumptions and the theory.


%
%
%
At the same sparsification ratio ($\theta$), we show FFT preserves more gradient information than other state-of-the-art lossy methodologies including Top-K sparsification, Terngrad, and QSGD. Besides, our adaptive float quantization further improves the overall compression ratio with negligible loss of gradient information (Fig.~\ref{gradient_dist_alg_comps}). These advantages enable us to use a larger compression ratio in retaining the same accuracy as the lossless SGD, than other lossy methodologies to improve the scalability (Fig.~\ref{scalability}) in the distributed training.

Our lossy gradient compression framework demands a highly efficient allreduce that supports communications of sparse data, while current MPI implementations, such as Open MPI or MVAPICH, lack the support of sparse collectives. Though this work uses all-gather to circumvent this issue, future research and development of a bandwidth-efficient allreduce with the sparse support are highly desired to facilitate the deployment of lossy gradient compression techniques in practice.

\section{Acknowledgement}
The work done by George Bosilca is supported by the National Science Foundation project SI2-SSI: EVOLVE, under Award Number 1664142.


\bibliographystyle{ACM-Reference-Format}
\bibliography{hpdc}

\appendix

\newpage
\onecolumn
\begin{lemma}
	\label{lemma:descent_appendix}
	Assume $\eta_t\leq\frac{1}{4L}, \theta_t^2\leq \frac{1}{4}$. Then
	\be
	\label{lm:descent_appendix}
	\frac{\eta_t}{4}\E[\|\nabla f(\xt)\|^2] \leq \E[f(\xt)] - \E[f(\xte)] + (L\eta_t + \theta_t^2)\frac{\eta_t\sigma^2}{2b_t}.
	\ee
\end{lemma}

\begin{proof}
	By Lipschitz continuity,
	\beaa
	f(\xte) & = & f(\xt - \eta_t\hat{v}_t)\\
	& \leq & f(\xt) + \langle \nabla f(\xt), -\eta_t\hat{v}_t\rangle + \frac{L}{2}\|\eta_t\hat{v}_t\|^2 \\
	& = & f(\xt) + \langle \nabla f(\xt), -\eta_t v_t\rangle + \langle \nabla f(\xt), -\eta_t(\hat{v}_t-v_t)\rangle + \frac{L}{2}\|\eta_t\hat{v}_t\|^2\\
	& \leq & f(\xt) + \langle \nabla f(\xt), -\eta_t v_t\rangle + \eta_t\|\nabla f(\xt)\|\|\hat{v}_t-v_t\| + \frac{L}{2}\|\eta_t\hat{v}_t\|^2\\
	& \leq & f(\xt) + \langle \nabla f(\xt), -\eta_t v_t\rangle + \eta_t\|\nabla f(\xt)\|\|\hat{v}_t-v_t\| + \frac{L}{2}\eta_t^2\|v_t\|^2\\
	& \leq & f(\xt) + \langle \nabla f(\xt), -\eta_t v_t\rangle + \frac{\eta_t}{2}\|\nabla f(\xt)\|^2 + \frac{\eta_t}{2}\|\hat{v}_t-v_t\|^2 + \frac{L}{2}\eta_t^2\|v_t\|^2\\
	& \leq & f(\xt) + \langle \nabla f(\xt), -\eta_t v_t\rangle + \frac{\eta_t}{2}\|\nabla f(\xt)\|^2 + \frac{\eta_t\theta_t^2}{2}\|v_t\|^2 + \frac{L}{2}\eta_t^2\|v_t\|^2\\
	& = & f(\xt) + \langle \nabla f(\xt), -\eta_t v_t\rangle + \frac{\eta_t}{2}\|\nabla f(\xt)\|^2 + \frac{\eta_t}{2}(L\eta_t + \theta_t^2)\|v_t\|^2\\
	\eeaa
	Note that conditioning on $\xt$, $\E[v_t|\xt] = \nabla f(\xt)$. Then take expectation on both sides, and use the relationship  that
	\beaa
	\E[\|v_t\|^2|\xt] & = & \E[\|v_t-\nabla f(x)\|^2|\xt] + \|\E[v_t|\xt]\|^2\\
	& = & \E[\|v_t-\nabla f(x)\|^2|\xt] +\|\nabla f(\xt)\|^2\\
	&\leq & \frac{\sigma^2}{b_t} + 	\|\nabla f(\xt)\|^2
	\eeaa
	We have
	\beaa
	\E[f(\xte)|\xt]& \leq & f(\xt) + \langle \nabla f(\xt), -\eta_t\E[v_t|\xt]\rangle + \frac{\eta_t}{2}\|\nabla f(\xt)\|^2 + \frac{\eta_t}{2}(L\eta_t + \theta_t^2)\E[\|v_t\|^2|\xt]\\
	& \leq & f(\xt) - \frac{\eta_t}{2}\|\nabla f(\xt)\|^2 + \frac{\eta_t}{2}(L\eta_t + \theta_t^2)(\frac{\sigma^2}{b_t} + \|\nabla f(\xt)\|^2)\\
	& = & f(\xt) - \frac{\eta_t}{2}(1-L\eta_t - \theta_t^2)\|\nabla f(\xt)\|^2 +  (L\eta_t + \theta_t^2)\frac{\eta_t\sigma^2}{2b_t}
	\eeaa
	Note that $L\eta_t\leq 1/4, \theta_t^2\leq 1/4$, and take expectation over the history, we get
	\beaa
	\E[f(\xte)]& \leq & \E[f(\xt)] - \frac{\eta_t}{4}\E[\|\nabla f(\xt)\|^2] +  (L\eta_t + \theta_t^2)\frac{\eta_t\sigma^2}{2b_t}
	\eeaa
	The lemma is proved.
\end{proof}

\end{document}